\documentclass[sigconf]{acmart}

\usepackage{multirow}
\usepackage{appendix}
\usepackage{subfig}
\usepackage{graphicx}
\usepackage{grffile}

\newcounter{subcopyrightbox@save}
\usepackage[font=bf]{caption}
\usepackage{color, url}
\usepackage{xspace} 
\usepackage{mathrsfs}
\usepackage{amsmath}
\usepackage{amsthm}
\usepackage{epstopdf}
\usepackage{balance}
\usepackage{bm}
\usepackage{rotating}
\usepackage{enumitem}
\usepackage{makecell}
\setlist{nolistsep}
\usepackage{algorithm}
\usepackage{algorithmic}
\usepackage{moresize}
\allowdisplaybreaks

\usepackage{thm-restate,thmtools}

\newtheorem{theorem}{Theorem}
\newtheorem{lemma}[theorem]{Lemma}

\newcommand{\argmax}{\operatornamewithlimits{argmax}}
\newcommand{\argmin}{\operatornamewithlimits{argmin}}
\newcommand{\myparatight}[1]{\smallskip\noindent{\bf {#1}:}~}

\makeatletter

\newenvironment{customLemma}[1]
  {\count@\c@lemma
   \global\c@lemma#1 %
    \global\advance\c@lemma\m@ne
   \lemma}
  {\endlemma
   \global\c@lemma\count@}

\makeatother

\graphicspath{ {../figs/} }

\AtBeginDocument{%
  \providecommand\BibTeX{{%
    \normalfont B\kern-0.5em{\scshape i\kern-0.25em b}\kern-0.8em\TeX}}}

\copyrightyear{2021}
\acmYear{2021}
\setcopyright{acmlicensed}
\acmConference[KDD '21]{Proceedings of the 27th ACM SIGKDD Conference on Knowledge Discovery and Data Mining}{August 14--18, 2021}{Virtual Event, Singapore}
\acmBooktitle{Proceedings of the 27th ACM SIGKDD Conference on Knowledge Discovery and Data Mining (KDD '21), August 14--18, 2021, Virtual Event, Singapore}
\acmPrice{15.00}
\acmDOI{10.1145/3447548.3467295} 
\acmISBN{978-1-4503-8332-5/21/08}

\begin{document}

\fancyhead{}

\title{Certified Robustness of Graph Neural Networks against Adversarial Structural Perturbation}

\author{Binghui Wang, Jinyuan Jia, Xiaoyu Cao, and Neil Zhenqiang Gong}
\affiliation{%
  \institution{ECE Department, Duke University, Durham, NC, USA, 27708}
}
\email{{binghui.wang,jinyuan.jia,xiaoyu.cao,neil.gong}@duke.edu}

\begin{abstract}
Graph neural networks (GNNs) have recently gained much attention for node and graph classification tasks on graph-structured data. 
However, multiple recent works showed that an attacker can easily make GNNs predict incorrectly via perturbing the graph structure, i.e., adding or deleting edges in the graph. We aim to defend against such attacks via developing certifiably robust GNNs. Specifically, we prove the certified robustness guarantee of any GNN for both node and graph classifications against structural perturbation. Moreover, we show that our certified robustness guarantee is \emph{tight}. Our results are based on a recently proposed technique called \emph{randomized smoothing}, which we extend to graph data. We also empirically evaluate our method for both node and graph classifications on multiple GNNs and multiple benchmark datasets. For instance, on the Cora dataset, Graph Convolutional Network with our randomized smoothing can achieve a certified accuracy of 0.49 when the attacker can arbitrarily add/delete at most 15 edges in the graph.
\end{abstract}



\maketitle

\section{Introduction}
\label{intro}

Graphs are a powerful tool to represent data from diverse domains such as social networks, biology, finance, security, etc.. 
\emph{Node classification} and \emph{graph classification} are two basic tasks on graphs. Specifically, given a graph, node classification aims to classify nodes in the graph in a semi-supervised fashion, while graph classification aims to assign a label to the entire graph instead of individual nodes. Node and graph classifications have many applications, including but not limited to, profiling users in social networks~\cite{zheleva2009join,mislove2010you,gong2016you,jia2017attriinfer}, classifying proteins~\cite{hamilton2017inductive,velivckovic2018graph}, and fraud detection~\cite{pandit2007netprobe,tamersoy2014guilt,wang2017gang,weber2019anti,gong2014sybilbelief,wang2017sybilscar,wang2018structure,wang2019graph}. 
Various methods such as label propagation~\cite{zhu2003semi},  belief propagation~\cite{pearl1988probabilistic}, iterative classification~\cite{sen2008collective}, and graph neural networks (GNNs)~\cite{kipf2017semi,hamilton2017inductive,gilmer2017neural,velivckovic2018graph,xu2018representation} have been proposed for node/graph classifications. Among them, GNNs have attracted much attention recently because of their expressiveness and superior performance.

However, several recent works~\cite{zugner2018adversarial,dai2018adversarial,zugner2019adversarial,bojchevski2019adversarial,wang2019attacking} have demonstrated that an attacker can easily fool GNNs to make incorrect predictions via perturbing 1) the node features and/or 2) the structure of the graph (i.e., adding and/or deleting edges in the graph). Therefore, it is of great importance to develop certifiably robust GNNs against such attacks. 
A few recent works~\cite{Zugner2019Certifiable,bojchevski2020efficient,jin2020certified,bojchevski2019certifiable,zugner2020certifiable}
study certified robustness of GNNs against node feature or structural perturbation.
However, these methods only focus on certifying robustness for a \emph{specific} GNN~\cite{Zugner2019Certifiable,bojchevski2019certifiable,jin2020certified,zugner2020certifiable} or/and did not formally prove \emph{tight} robustness guarantees~\cite{bojchevski2020efficient} in general. 

We aim to bridge this gap in this work. Specifically, we aim to provide \emph{tight} certified robustness guarantees of \emph{any} GNN classifier against structural perturbations for both node and graph classification. Towards this goal, we leverage a recently developed  technique called \emph{randomized smoothing}~\cite{cao2017mitigating,liu2018towards,lecuyer2018certified,li2018second,cohen2019certified,jia2020certified}, which can turn any \emph{base classifier} (e.g., a GNN classifier in our problem) to a robust one via adding random noise to a testing example (i.e., a graph in our problem). 
A classifier is certifiably robust if it provably predicts the same label  when the attacker adds/deletes at most $K$ edges in the graph, where we call $K$ \emph{certified perturbation size}.

Graph structures are \emph{binary} data, i.e., a pair of nodes can be either connected or unconnected. 
Therefore, we develop randomized smoothing for binary data and leverage it to obtain certified robustness of GNNs against structural perturbation. 
First, we theoretically derive a certified perturbation size for any GNN classifier with randomized smoothing via addressing several  
challenges. For instance, we divide the graph structure space into regions in a novel way such that we can apply the Neyman-Pearson Lemma~\cite{neyman1933ix} to certify robustness. 
We also prove that our derived certified perturbation size is \emph{tight} if no assumptions on the GNN classifier are made. 

Second, we design a method to compute our certified perturbation size in practice.
It is challenging to compute our certified perturbation size as it involves estimating probability bounds \emph{simultaneously} and solving an optimization problem.  
To address the challenge, we first adopt the simultaneous confidence interval estimation method~\cite{jia2020certified} to estimate the probability bounds with probabilistic guarantees.  
Then, we design an algorithm to solve the optimization problem to obtain our certified perturbation size with the estimated probability bounds.

We also empirically evaluate our method. Specifically, for node classification, we consider  Graph Convolutional Network (GCN)~\cite{kipf2017semi} and Graph Attention Network (GAT)~\cite{velivckovic2018graph} on  several benchmark datasets including Cora, Citeseer, and Pubmed~\cite{sen2008collective}. For graph classification, we consider Graph Isomorphism Network (GIN)~\cite{xu2019powerful} on benchmark datasets including MUTAG, PROTEINS, and IMDB~\cite{yanardag2015deep}.  For instance, on the Cora dataset, GCN with our randomized smoothing can achieve certified accuracies of 0.55, 0.50, and 0.49 when the attacker arbitrarily adds/deletes at most 5, 10, and 15 edges, respectively. On the MUTAG dataset, GIN with our randomized smoothing can achieve certified accuracies of 0.45, 0.45, and 0.40 when the attacker arbitrarily adds/deletes at most 5, 10, and 15 edges, respectively. 

Our major contributions can be 
summarized as follows:
\begin{itemize}
\item We prove the certified robustness guarantee of any GNN against structural perturbation. Moreover, we show that our certified robustness guarantee is tight.
\item Our certified perturbation size is the solution to an optimization problem and we design an algorithm to solve the 
problem.
\item We empirically evaluate our method for both node and graph classifications on multiple benchmark datasets.
\end{itemize}

\section{Background}
\subsection{Node Classification vs. Graph Classification} 
We consider GNNs for both \emph{node classification}~\cite{kipf2017semi,hamilton2017inductive,velivckovic2018graph} and \emph{graph classification}~\cite{hamilton2017inductive,xu2019powerful}. Suppose we are given an undirected graph $G$ with node features. 
\begin{itemize}[leftmargin=*]
\item {\bf Node classification.} A node classifier predicts labels for nodes in the graph in a semi-supervised fashion. Specifically, a subset of nodes in $G$ are already labeled, which are called \emph{training nodes} and denoted as $V_T$. A node classifier $f_n$ takes the graph $G$ and the training nodes $V_T$ as an input and predicts the label for 
remaining nodes, i.e., $f_n(G, V_T, u)$ is the predicted label for 
a 
node $u$. 
\item  {\bf Graph classification.} A graph classifier aims to predict a label for the entire graph instead of individual nodes, i.e., $f_g(G)$ is the predicted label for the graph $G$. Such a graph classifier can be trained using a set of graphs with ground truth labels. 
\end{itemize}

\subsection{Adversarial Structural Perturbation} 
An attacker aims to fool a node classifier or graph classifier to make predictions as the attacker desires via perturbing the graph structure, i.e., deleting some edges and/or adding some edges in the graph~\cite{zugner2018adversarial,dai2018adversarial,zugner2019adversarial,bojchevski2019adversarial,wang2019attacking}. Since our work focuses on structural perturbation,  we treat the node feature matrix and the training nodes as constants. Moreover, we simply write a node classifier $f_n(G, V_T, u)$  or a graph classifier $f_g(G)$ as $f(\mathbf{s})$,  where the binary vector $\mathbf{s}$ represents the graph structure. Note that a node classifier should also take a node $u$ as input and predict its label. However, we omit the explicit dependency on a node $u$ for simplicity. We call $\mathbf{s}$ \emph{structure vector}.  For instance, $\mathbf{s}$ can be the concatenation of the upper triangular part of the adjacency matrix of the graph (excluding the diagonals) when the attacker can modify the connection status of any pair of nodes, i.e., $\mathbf{s}$ includes the connection status for each pair of nodes in the graph.  When the attacker can only modify the connection status between $u$ and each remaining node, $\mathbf{s}$ can also be the $u$th row of the adjacency matrix of the graph (excluding the self-loop $(u,u)$th entry). We assume the structure vector  $\mathbf{s}$ has $n$ entries. As we will see, such simplification makes it easier to present our certified robustness against structural perturbation.

We denote by vector $\delta\in \{0,1\}^n$ the attacker's perturbation to the graph structure. Specifically,  $\delta_i=1$ if and only if the attacker changes the $i$th entry in the structure vector  $\mathbf{s}$, i.e., the attacker changes the connection status of the corresponding pair of nodes. Moreover, $\mathbf{s} \oplus \delta$ is the perturbed structure vector, which represents the perturbed graph structure, where $\oplus$ is the XOR operator between two binary variables. We use $||\delta||_0$ to measure the magnitude of the adversarial perturbation as it has semantic meanings. Specifically, $||\delta||_0$ is the number of node pairs whose connection statuses are modified by the attacker.

\section{Certified Robustness}
\label{certify}

\subsection{Randomized Smoothing with Binary Noise}
We first define a noise distribution in the discrete structure vector space $\{0,1\}^n$. Then, we define a \emph{smoothed classifier} based on the noise distribution and a  node/graph classifier (called \emph{base classifier}). Specifically, we consider the noise vector $\epsilon$ has the following probability distribution in the discrete space $\{0,1\}^n$:
\begin{align}
\label{discretenoisedistribution}
\text{Pr}(\epsilon_i=0) =\beta,\ \text{Pr}(\epsilon_i=1) =1-\beta,
\end{align}
where $i=1,2,\cdots,n$. 
When we add a random noise vector $\epsilon$ to the structure vector $\mathbf{s}$, the $i$th entry of $\mathbf{s}$ is preserved with probability $\beta$ and changed with probability $1-\beta$. In other words, our random noise means that, for each pair of nodes in the graph, we keep its connection status with probability $\beta$ and change its connection status with probability $1-\beta$.  Based on the noise distribution and a base node/graph classifier $f(\mathbf{s})$, we define a smoothed classifier $g(\mathbf{s})$ as follows:
\begin{align}
\label{smoothedclassifier}
g(\mathbf{s}) = \argmax_{c\in \mathcal{C}}   \text{Pr}(f(\mathbf{s}\oplus \epsilon)=c),
\end{align}
where $\mathcal{C}$ is the set of labels, $\oplus$ is the XOR operator between two binary variables, $\text{Pr}(f(\mathbf{s}\oplus \epsilon)=c)$ is the probability that the base classifier $f$ predicts label $c$ when we add random noise $\epsilon$ to the structure vector $\mathbf{s}$, and $g(\mathbf{s})$ is the label predicted for $\mathbf{s}$ by the smoothed classifier. 
Moreover, we note that $g(\mathbf{s} \oplus \delta)$ is the label predicted for the perturbed structure vector $\mathbf{s} \oplus \delta$.
Existing randomized smoothing methods~\cite{cao2017mitigating,liu2018towards,lecuyer2018certified,li2018second,cohen2019certified,jia2020certified} add random continuous noise (e.g., Gaussian noise, Laplacian noise) to a testing example (i.e., the structure vector $\mathbf{s}$ in our case). However, such continuous noise is not meaningful for the binary structure vector.

Our goal is to show that a label is provably the predicted label by the smoothed classifier $g$ for the perturbed structure vector $\mathbf{s} \oplus \delta$ when the $\ell_0$-norm of the adversarial perturbation $\delta$, i.e., $\|\delta\|_0$, is bounded. Next, we theoretically derive the certified perturbation size $K$, prove the tightness of the certified perturbation size, and discuss how to compute the certified perturbation size in practice. 

\subsection{Theoretical Certification} 

\subsubsection{Overview}

Let $X = \mathbf{s} \oplus \epsilon$ and $Y=\mathbf{s} \oplus \delta \oplus \epsilon$ be two random variables, 
where $\epsilon$ is the random binary noise drawn from the distribution defined in Equation~\ref{discretenoisedistribution}. 
$X$ and $Y$ represent random structure vectors obtained by adding random binary noise $\epsilon$ to the structure vector $\mathbf{s}$ and its perturbed version $\mathbf{s} \oplus \delta$, respectively.

Suppose, when taking $X$ as an input, the base GNN classifier $f$ correctly predicts the label $c_A$ with the largest probability.
Then, our key idea is to guarantee that, when taking $Y$ as an input, $f$ still predicts $c_A$ with the largest probability. 
Moreover, we denote $c_B$ as the predicted label by $f$ with the second largest probability.
 
Then, our goal is to find the maximum perturbation size such that the following inequality holds:
\begin{align}
\label{bound_both_sides}
\text{Pr}(f(Y)=c_A)> \text{Pr}(f(Y)=c_B). 
\end{align}
Note that it is challenging to compute the probabilities $\text{Pr}(f(Y) = c_A)$ and $\text{Pr}(f(Y)=c_B)$ exactly because $f$ is highly nonlinear in practice.
To address the challenge, we first derive a lower bound of $\text{Pr}(f(Y) = c_A)$ and an upper bound of $\text{Pr}(f(Y)=c_B)$. 
Then, we require that the lower bound of $\text{Pr}(f(Y) = c_A)$ is larger than the upper bound of $\text{Pr}(f(Y)=c_B)$.
Specifically, we derive the lower bound and upper bound by constructing certain regions in the graph structure space  $\{0,1\}^n$ such that the probabilities $Y$ is in these regions can be efficiently computed for any $||\delta||_0=k$. 
Then, we iteratively search the maximum $k$ under the condition that the lower bound is larger than the upper bound. 
Finally, we treat the maximum $k$ as the certified perturbation size $K$.

\subsubsection{Deriving the lower and upper bounds}
Our idea is to divide the graph structure space  $\{0,1\}^n$ into regions in a novel way such that we can apply the Neyman-Pearson Lemma~\cite{neyman1933ix} to derive the  lower bound and upper bound. 
First, for any data point $ \mathbf{z} \in \{0,1\}^n$, we have the density ratio $\frac{\text{Pr}(X = \mathbf{z})}{\text{Pr}(Y = \mathbf{z})}=\big(\frac{\beta}{1-\beta}\big)^{w-v}$ based on the noise distribution defined in Equation~\ref{discretenoisedistribution}, where $ w=\| \mathbf{s} - \mathbf{z} \|_0$ and $v = \| \mathbf{s} \oplus \delta - \mathbf{z} \|_0$  (please refer to Section~\ref{compute_prob} in for details). 
Therefore, we have the density ratio $\frac{\text{Pr}(X = \mathbf{z})}{\text{Pr}(Y = \mathbf{z})}=\big(\frac{\beta}{1-\beta}\big)^{m}$ for any $ \mathbf{z} \in \{0,1\}^n$, where $m=-n, -n+1, \cdots, n-1, n$. 
Furthermore, we define a region $\mathcal{R}(m)$ as the set of data points whose density ratios are $\big(\frac{\beta}{1-\beta}\big)^{m}$, i.e.,  $\mathcal{R}(m)= \{\mathbf{z} \in \{0,1\}^n: \frac{\text{Pr}(X = \mathbf{z})}{\text{Pr}(Y = \mathbf{z})} = \big(\frac{\beta}{1-\beta}\big)^{m}\}$, and we denote by $r(m)$ the corresponding density ratio, i.e.,  $r(m) = \big(\frac{\beta}{1-\beta}\big)^{m}$. 
Moreover, we rank the $2n+1$ regions  in a descending order with respect to their density ratios, and denote the ranked regions as $\mathcal{R}_1, \mathcal{R}_2, \cdots, \mathcal{R}_{2n+1}$.

Suppose we have
a lower bound of the largest label probability $\text{Pr}(f(X)=c_A)$ for $c_A$ and denote it as $\underline{p_A}$, 
and an upper bound of the remaining label probability $\text{Pr}(f(X)=c)$ for $c \neq c_A$ and denote it as $\overline{p_B}$. 
Assuming there exist $c_A$ and $\underline{p_A}$, $\overline{p_B} \in [0,1]$ such that  
\begin{align}
\label{main_theorem_condition_label}
\text{Pr}(f(X)=c_A)\geq \underline{p_A}\geq \overline{p_B} \geq \max_{c\neq c_A}\text{Pr}(f(X)=c).
\end{align}
Next, we construct two regions $\mathcal{A}$ and $\mathcal{B}$ such that $\text{Pr}(X \in \mathcal{A}) = \underline{p_A}$ and $\text{Pr}(X \in \mathcal{B}) = \overline{p_B}$, respectively. 
Specifically, we gradually add the regions $\mathcal{R}_1, \mathcal{R}_2, \cdots, \mathcal{R}_{2n+1}$ to $\mathcal{A}$ until  $\text{Pr}(X \in \mathcal{A}) = \underline{p_A}$. 
Moreover, we gradually add the regions $\mathcal{R}_{2n+1}, \mathcal{R}_{2n}, \cdots, \mathcal{R}_{1}$ to $\mathcal{B}$ until $\text{Pr}(X \in \mathcal{B}) = \overline{p_B}$.
We construct the regions $\mathcal{A}$ and $\mathcal{B}$ in this way such that we can apply the Neyman-Pearson Lemma~\cite{neyman1933ix} for them.
Formally, we define the regions $\mathcal{A}$ and $\mathcal{B}$ as follows:
\begin{align}
& \mathcal{A} = \bigcup_{j=1}^{a^{\star}-1} \mathcal{R}_j \cup \underline{\mathcal{R}_{a^\star}} \label{region_A} \\
& \mathcal{B} = \bigcup_{j=b^{\star}+1}^{2n+1} \mathcal{R}_j \cup \underline{\mathcal{R}_{b^\star}}, \label{region_B}
\end{align}
where 
\begin{align*}
& a^{\star} = \argmin_{a \in \{1, 2, \cdots, 2n+1 \}} a, \ s.t.\ \sum_{j=1}^{a} \textrm{Pr}(X \in \mathcal{R}_j) \geq \underline{p_A}, \\
& b^{\star} = \argmax_{b \in \{1, 2, \cdots, 2n+1 \}} b,\ s.t.\ \sum_{j=b}^{2n+1} \textrm{Pr}(X \in \mathcal{R}_j) \geq \overline{p_B}.
\end{align*}
$\underline{\mathcal{R}_{a^{\star}}}$ is any subregion of $\mathcal{R}_{a^{\star}}$ such that
$\text{Pr}(X\in \underline{\mathcal{R}_{a^{\star}}})= \underline{p_A} - \sum_{j=1}^{a^{*}-1} \textrm{Pr}(X \in \mathcal{R}_j)$, and  
$\underline{\mathcal{R}_{b^{\star}}}$ is any subregion of  $\mathcal{R}_{b^{\star}}$ such that
$\text{Pr}(X\in \underline{\mathcal{R}_{b^{\star}}})= \overline{p_B} - \sum_{j={b^*+1}}^{2n+1} \textrm{Pr}(X \in \mathcal{R}_j)$.

Finally, based on the Neyman-Pearson Lemma, we can derive a lower bound of $\text{Pr}(f(Y)=c_A)$ and an upper bound of $\text{Pr}(f(Y)=c_B)$. 
Formally, we have the following lemma: 
\begin{restatable}[]{lem}{certifyiedradiuslemma} 
\label{lemmabounds} 
We have the following bounds:
\begin{align}
\text{Pr}(f(Y)=c_A) \geq \text{Pr}(Y\in\mathcal{A}) \label{lower}, \\ 
\text{Pr}(f(Y)=c_B) \leq \text{Pr}(Y\in\mathcal{B}) \label{upper}.
\end{align}
\end{restatable}
\begin{proof}
See Section~\ref{proof_lemmabounds}. 
\end{proof}

\subsubsection{Deriving the certified perturbation size}
Given Lemma~\ref{lemmabounds}, we can derive the certified perturbation size $K$ as the maximum $k$ such that the following inequality holds for $\forall ||\delta||_0=k$:
\begin{align}
\label{suff_condi}
\text{Pr}(Y \in \mathcal{A}) > \text{Pr}(Y\in\mathcal{B}).
\end{align}
Formally, we have the following theorem:
\begin{restatable}[Certified Perturbation Size]{thm}{certyfiedperturbationsize}
\label{CertifiedPerturbationSize}
Let $f$ be any base node/graph classifier. The random noise vector $\epsilon$  is defined in Equation~\ref{discretenoisedistribution}. The smoothed classifier $g$ is defined in Equation~\ref{smoothedclassifier}. Given a structure vector $\mathbf{s} \in \{0, 1\}^n$, suppose there exist $c_A$ and $\underline{p_A}, \overline{p_B} \in [0,1]$ that satisfy Equation~\ref{main_theorem_condition_label}. 
Then, we have
\begin{align}
g(\mathbf{s} \oplus \delta) = c_A, \forall   ||\delta||_0 \leq K, 
\end{align}
where the certified perturbation size $K$ is the solution to the following optimization problem: 
\begin{align}
\label{problemK}
& K= \argmax k,   \\
& \textrm{s.t. }  \text{Pr}(Y \in \mathcal{A}) > \text{Pr}(Y\in\mathcal{B}), \, \forall ||\delta||_0 = k \label{inequalityconstraint}.
\end{align}
\end{restatable}
\begin{proof}
See Section~\ref{proof_theorem}. 
\end{proof}

Next, we show that our certified perturbation size is tight in the following theorem. 

\begin{restatable}[Tightness of the Certified Perturbation Size]{thm}{tightnessofcertifiedperturbationsize}
\label{tighnessbound}
Assume $\underline{p_A} \geq \overline{p_B}$, $\underline{p_A}+\overline{p_B}\leq 1$, and $\underline{p_A}+ (|\mathcal{C}|-1)\cdot \overline{p_B} \geq 1$, where $|\mathcal{C}| $ is the number of labels. For any perturbation $\delta$ with $||\delta||_0 > K$, there exists a base classifier $f^{*}$ consistent with Equation~\ref{main_theorem_condition_label} such that $g(\mathbf{s} \oplus \delta)\neq c_A$ or there exist ties.  
\end{restatable}
\begin{proof}
See Section~\ref{proof_tight}.
\end{proof}

We have several observations on our major theorems. 
\begin{itemize}
\item Our theorems are applicable to any base node/graph classifier. Moreover, although we focus on classifiers on graphs, our theorems are applicable to any base classifier that takes binary features as input.  

\item Our certified perturbation size $K$ depends on $\underline{p_A}$, $\overline{p_B}$, and $\beta$. In particular, when the probability bounds $\underline{p_A}$ and $\overline{p_B}$ are tighter, our certified perturbation size $K$ is larger. We use the probability bounds $\underline{p_A}$ and $\overline{p_B}$ instead of their exact values, because it is challenging to compute $p_A$ and  $p_B$ exactly. 

\item When  no assumptions on the base classifier are made and randomized smoothing with the noise distribution defined in Equation~\ref{discretenoisedistribution} is used, it is impossible to certify a perturbation size larger than $K$. 
\end{itemize}

\subsection{Certification in Practice}
\label{certify_practice} 
Computing our certified perturbation size $K$ in practice faces two challenges. The first challenge is to estimate the probability bounds $\underline{p_A}$ and $\overline{p_B}$. The second challenge is to solve the optimization problem in Equation~\ref{problemK} to get $K$ with the given $\underline{p_A}$ and $\overline{p_B}$. We leverage the method developed in \cite{jia2020certified} to address the first challenge, and we develop an efficient algorithm to address the second challenge. 

\myparatight{Estimating $\underline{p_A}$ and $\overline{p_B}$} 
We view the probabilities $p_1, p_2, \cdots, p_{|\mathcal{C}|}$ as a multinomial distribution over the label set $\mathcal{C}$. If we sample
a noise $\epsilon$ from our noise distribution uniformly at random, then
$f(\mathbf{s} \oplus \epsilon)$ can be viewed as a sample from the multinomial distribution.
Then, estimating $\underline{p_A}$ and $\overline{p_c}$ for $c \in \mathcal{C} \setminus \{c_A\} $ is essentially a
one-sided simultaneous confidence interval estimation problem. 
In particular, we leverage the simultaneous confidence interval estimation
method developed in \cite{jia2020certified} to estimate these bounds with
a confidence level at least $1-\alpha$. Specifically, we sample $d$ random
noise, i.e., $\epsilon^1, \epsilon^2, \cdots, \epsilon^d$, from the noise distribution defined in
Equation~\ref{discretenoisedistribution}. We denote by $d_c$ the frequency of the label $c$ predicted by 
the base classifier for the $d$ noisy examples. Formally, we have
$d_c=\sum_{j=1}^d \mathbb{I} (f(\mathbf{s} \oplus \epsilon^j)=c)$ for each $c\in \mathcal{C}$ and $\mathbb{I}$ is the indicator function.
Moreover, we assume $c_A$ has the largest frequency and the smoothed classifier predicts $c_A$ as the label, i.e., $g(\mathbf{s}) = c_A$. According
to \cite{jia2020certified}, we have the following probability bounds with a
confidence level at least $1 - \alpha$:
\begin{align}
\label{compute_lower_bound_simuem}
 & \underline{p_A}=B\left(\frac{\alpha}{|\mathcal{C}|};d_{c_A},d-d_{c_A}+1 \right) \\
 \label{compute_upper_bound_simuem}
 &    \overline{p_c}=B\left(1-\frac{\alpha}{|\mathcal{C}|};d_c+1, d-d_c \right), \  \forall c \neq c_A,
\end{align}
where $B(q; u,w)$ is the $q$th quantile of a beta distribution with shape parameters $u$ and $w$.
Then, 
we estimate $\overline{p_B}$ as follows: 
\begin{align}
\overline{p_B}=\min\left(\max_{c \neq c_A} \overline{p_c}, 1 - \underline{p_A}\right).
\label{compute_upper_bound_pB}
\end{align}

\myparatight{Computing $K$}  
After estimating ${\underline{p_A}}$ and  $\overline{p_B}$, we solve the optimization problem in Equation~\ref{problemK} to obtain $K$. 
First, we have: 
{
\begin{align}
& \text{Pr}(Y \in \mathcal{A}) = \sum_{j=1}^{a^{\star}-1} \text{Pr} (Y \in \mathcal{R}_j) + (\underline{p_A} - \sum_{j=1}^{a^{\star}-1} \text{Pr} (X \in \mathcal{R}_j)) /  r_{a^{\star}} \label{x2y}, \\
& \text{Pr}(Y\in \mathcal{B}) = \sum_{j=b^{\star}+1}^{2n+1} \text{Pr} (Y \in \mathcal{R}_j) + (\overline{p_B} - \sum_{j=b^{\star}+1}^{2n+1} \text{Pr} (X \in \mathcal{R}_j)) /  r_{b^{\star}}, \label{x2y2} 
\end{align}
}%
where $r_{a^{\star}}$ and $r_{b^{\star}}$ are the density ratios in the regions $\mathcal{R}_{a^\star}$ and $\mathcal{R}_{b^\star}$, respectively. 
Therefore, the key to solve the optimization problem is to compute $\text{Pr}(X \in \mathcal{R}(m))$ and $\text{Pr}(Y \in \mathcal{R}(m))$ for each $m\in \{-n, -n+1, \cdots, n\}$ when $||\delta||_0=k$. Specifically, we have:
\begin{align}
\label{probabilityX}
&\text{Pr}(X \in \mathcal{R}(m)) = \sum_{j=\max\{0,m\}}^{\min\{n,n+m\}} \beta^{n-(j-m)} (1-\beta)^{(j-m)} \cdot t(m,j)\\
\label{probabilityY}
&\text{Pr}(Y \in \mathcal{R}(m)) =\sum_{j=\max\{0,m\}}^{\min\{n,n+m\}} \beta^{n-j} (1-\beta)^j \cdot t(m,j),
\end{align}
where $t(m,j)$ is defined as follows:
{
\begin{align}
t(m,j) = 
    \begin{cases}
      0,  &\textrm{ if } (m+k) \textrm{ mod } 2 \neq 0, \\
      0,  &\textrm{ if } 2j-m<k, \\
  {n-k \choose \frac{2j-m-k}{2}} {k \choose \frac{k-m}{2}}, &\textrm{ otherwise. }  
  \end{cases}
\end{align}
}%
See Section~\ref{compute_prob} for the details on obtaining Equation~\ref{probabilityX} and~\ref{probabilityY}. 
Then, we iteratively find the largest $k$ such that the constraint in Equation~\ref{inequalityconstraint} is satisfied.

\section{Evaluation}
\label{eval}

\begin{table}[!t]\renewcommand{\arraystretch}{1.3}
\centering
\caption{Dataset statistics.}
\centering
\addtolength{\tabcolsep}{-2pt}
\begin{tabular}{|c|c|c|c|c|} \hline 
\multicolumn{2}{|c|}{\bf Dataset} & {\small \#Nodes} & {\small \#Edges}  & {\small \#Classes}   \\ \hline
\multirow{3}{*}{\bf \makecell{Node \\ Classification}} & {\bf Cora} &  {\small 2,708} & {\small 5,429} & {\small 7} \\ \cline{2-5} 
& {\bf Citeseer} &  {\small 3,327} & {\small 4,732}  & {\small 6} \\ \cline{2-5} 
& {\bf Pubmed} &  {\small 19,717} & {\small 44,338} & {\small 3}  \\ \hline \hline
\multicolumn{2}{|c|}{\bf Dataset} & {\small \#Graphs} & {\small Ave.\#Nodes} & {\small \#Classes}   \\ \hline 
\multirow{3}{*}{\bf \makecell{Graph \\ Classification}} & {\bf MUTAG} &  {\small 188} & {\small 17.9}  & {\small 2} \\ \cline{2-5} 
& {\bf PROTEINS} &  {\small 1,113} & {\small  39.1}  & {\small 2} \\ \cline{2-5} 
& {\bf IMDB} &  {\small 1,500} & {\small 13.0}  & {\small 3}  \\ \hline
\end{tabular}
\label{dataset_stat}
\end{table}

\begin{figure*}[!t]
\center
\subfloat[Cora]{\includegraphics[width=0.28\textwidth]{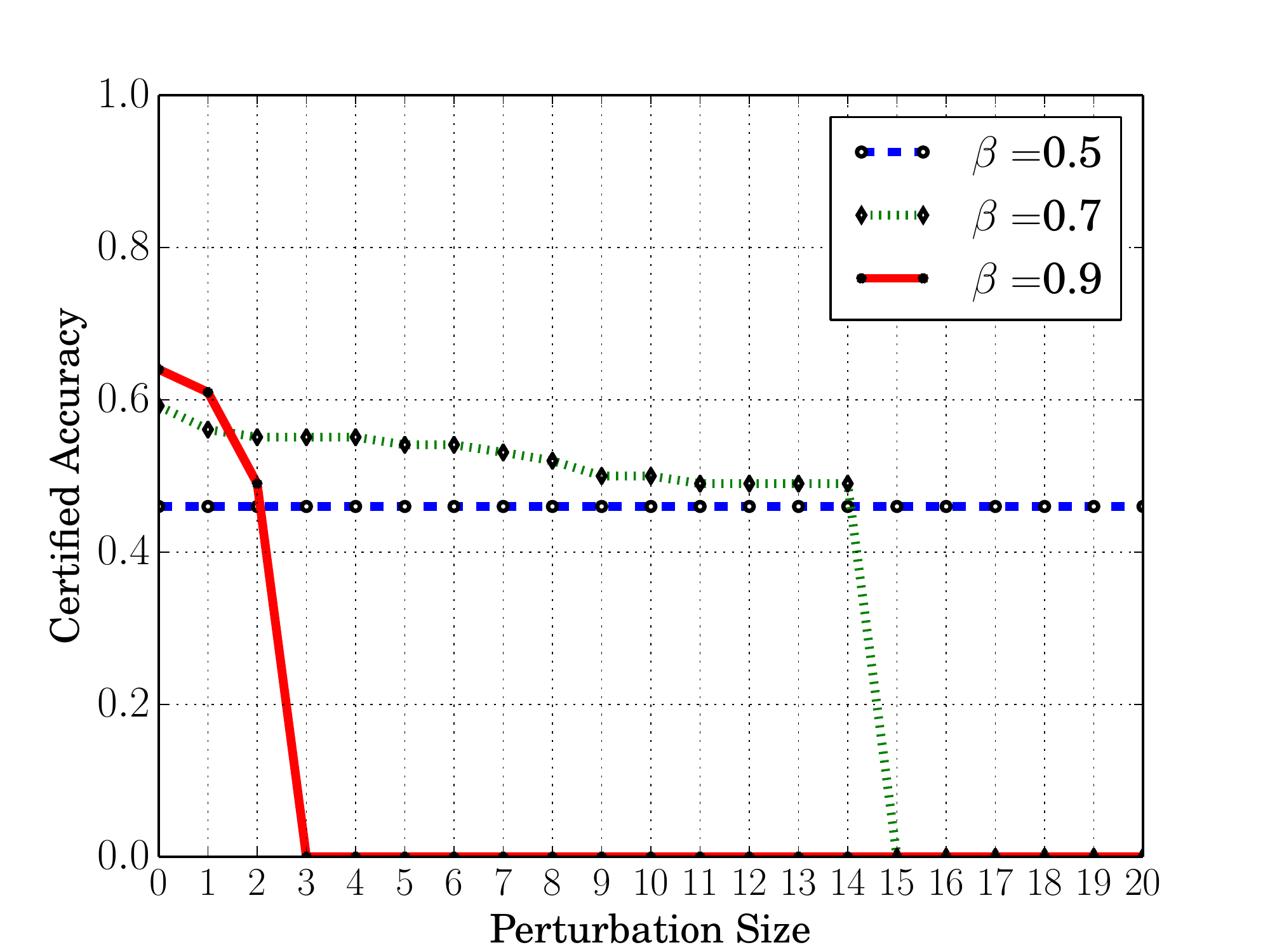} \label{fig1}} 
\subfloat[Citeseer]{\includegraphics[width=0.28\textwidth]{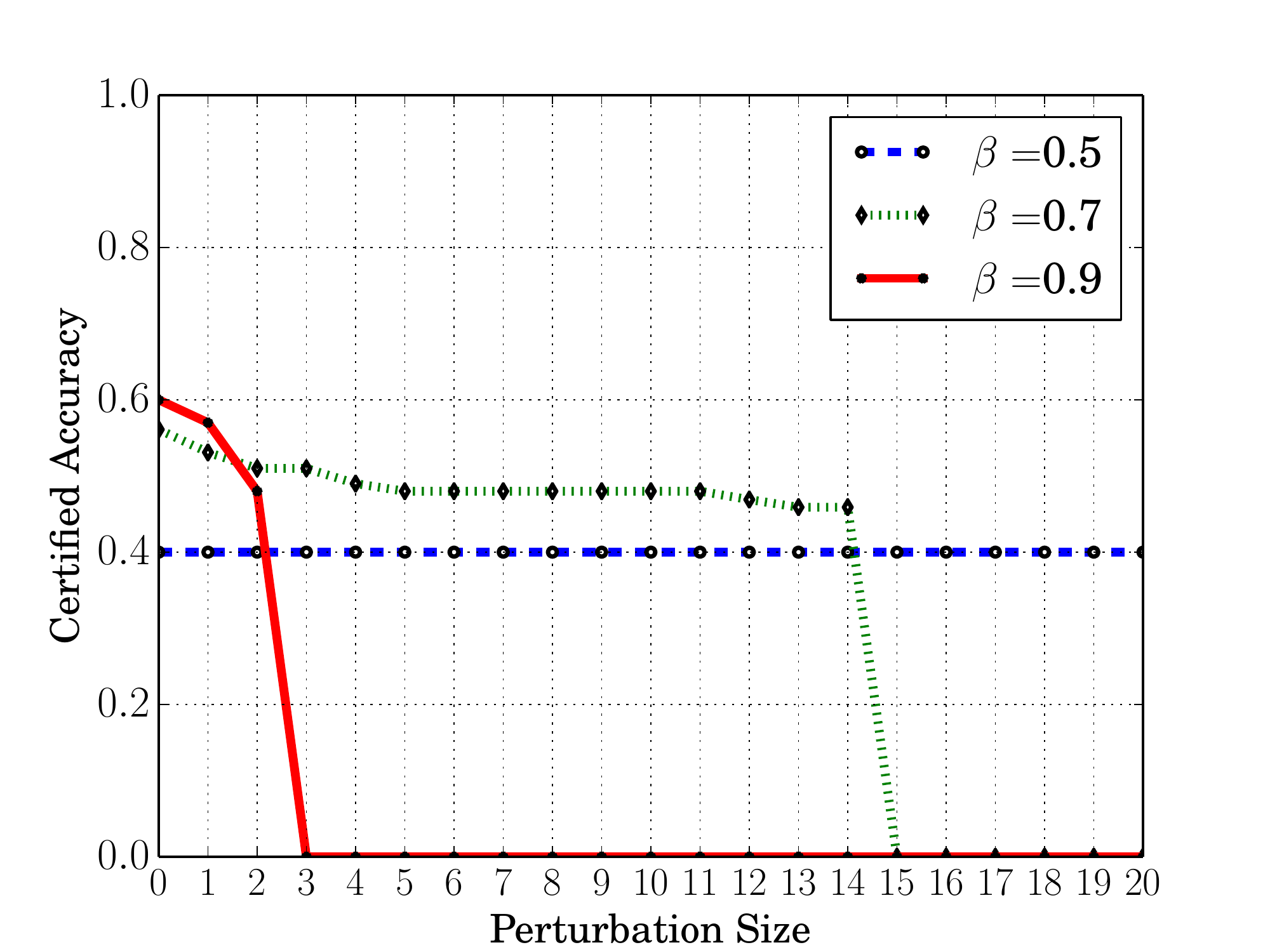} \label{fig2}} 
\subfloat[Pubmed]{\includegraphics[width=0.28\textwidth]{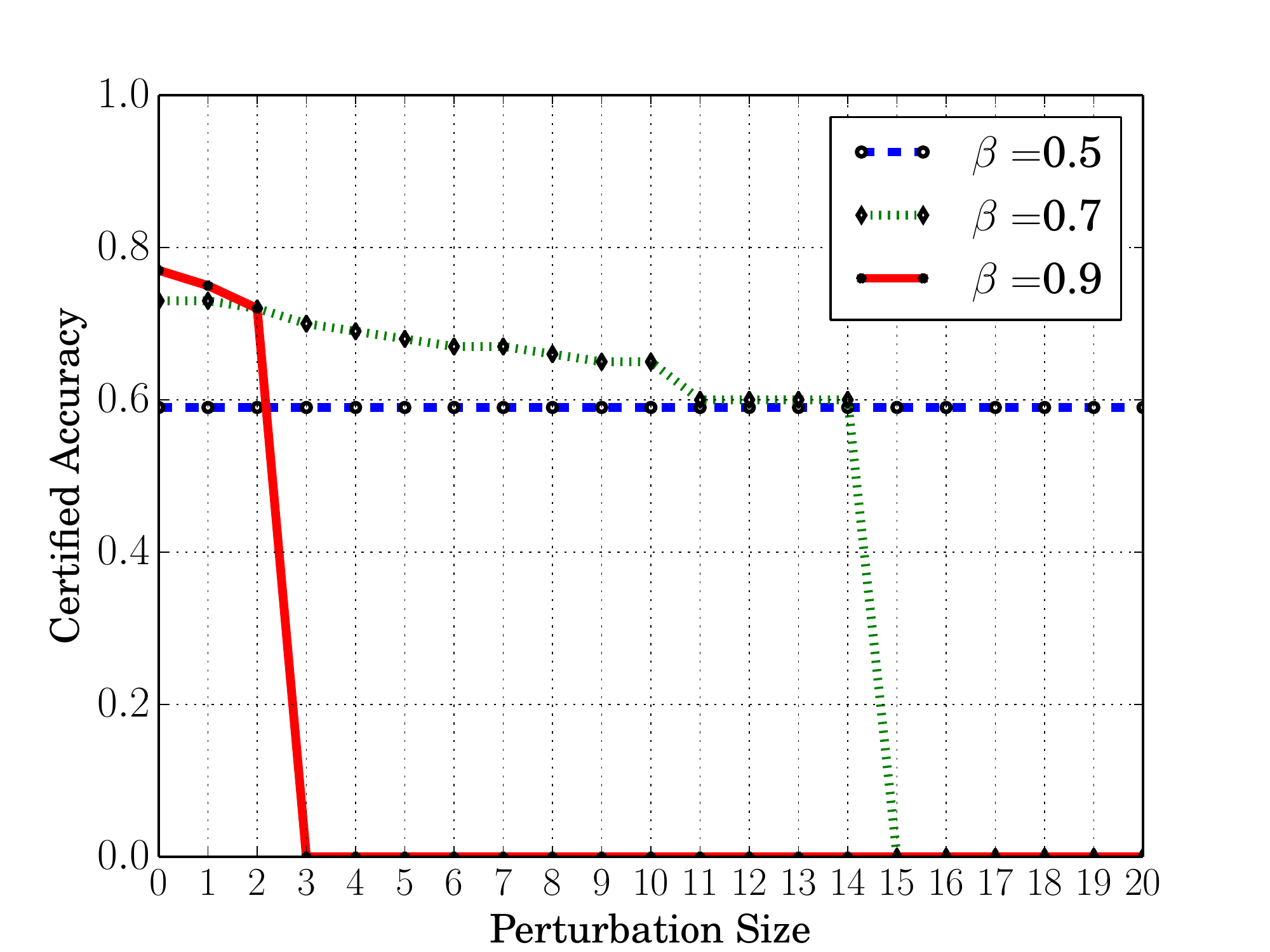} \label{fig3}} 
\vspace{-2mm}
\caption{Impact of $\beta$ on the certified accuracy of the smoothed GCN.}
\label{certify_GCN}
\vspace{-8mm}
\end{figure*}

\begin{figure*}[!t]
\center
\subfloat[Cora]{\includegraphics[width=0.28\textwidth]{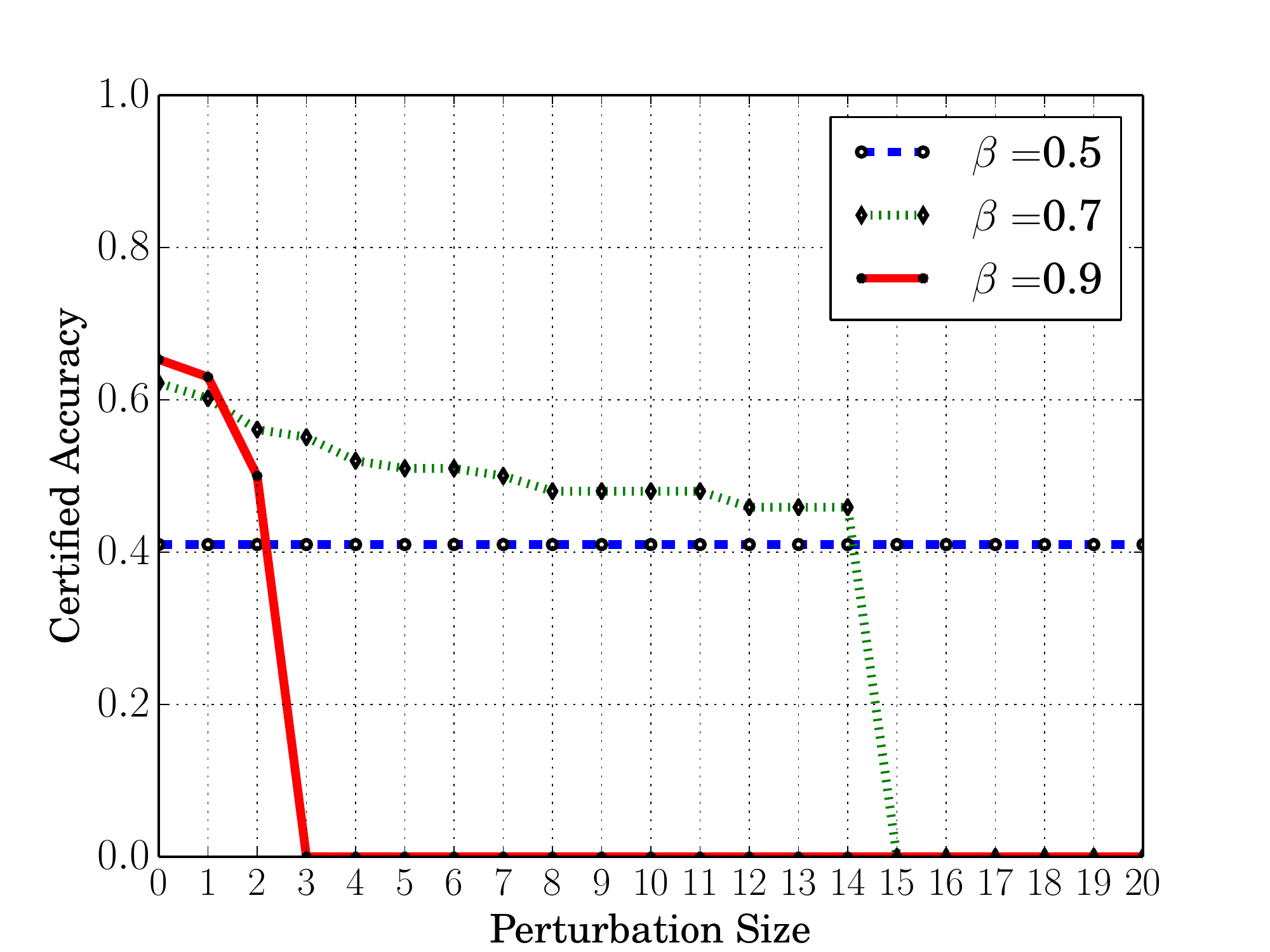} \label{fig1}} 
\subfloat[Citeseer]{\includegraphics[width=0.28\textwidth]{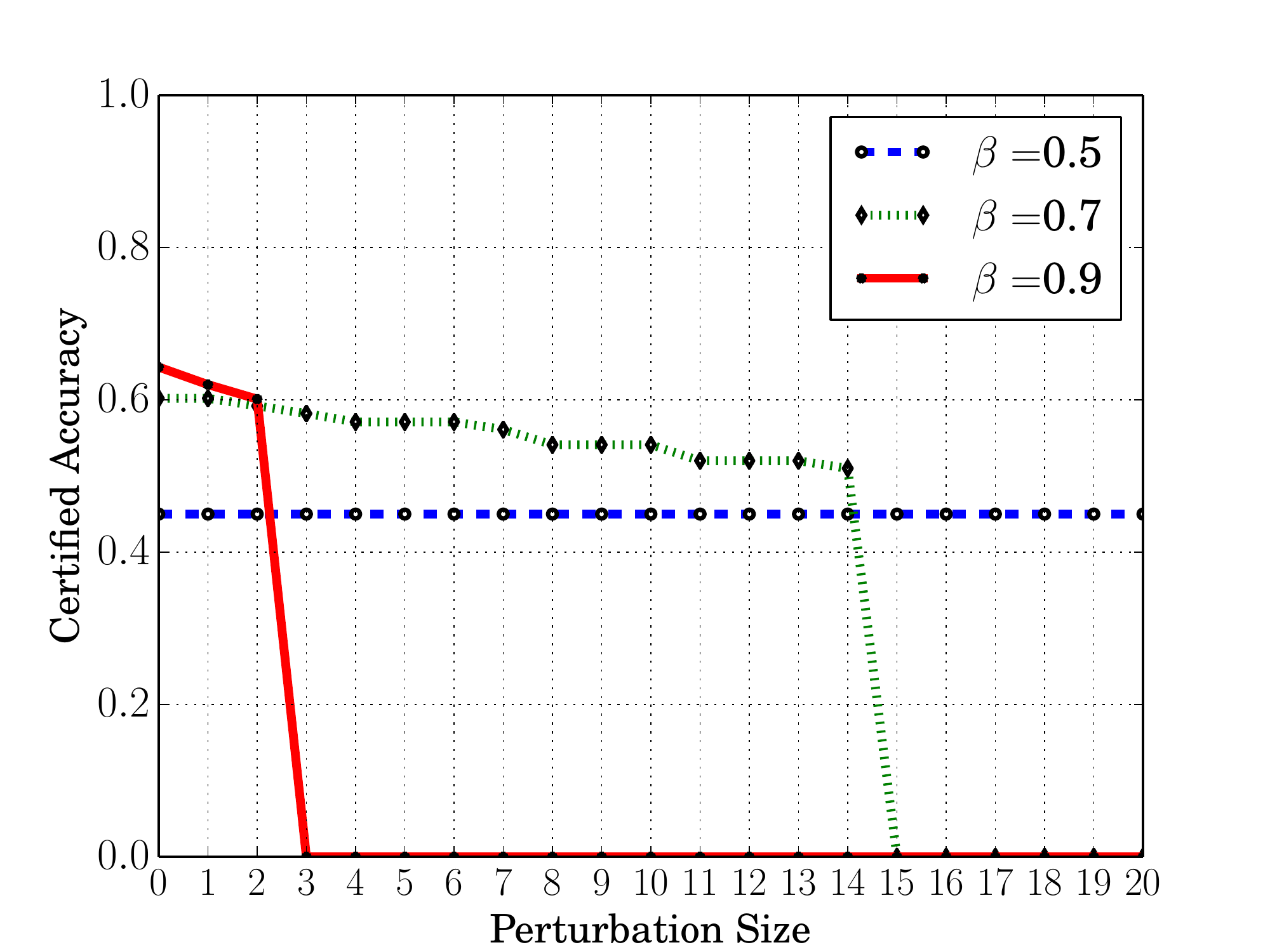} \label{fig2}} 
\subfloat[Pubmed]{\includegraphics[width=0.28\textwidth]{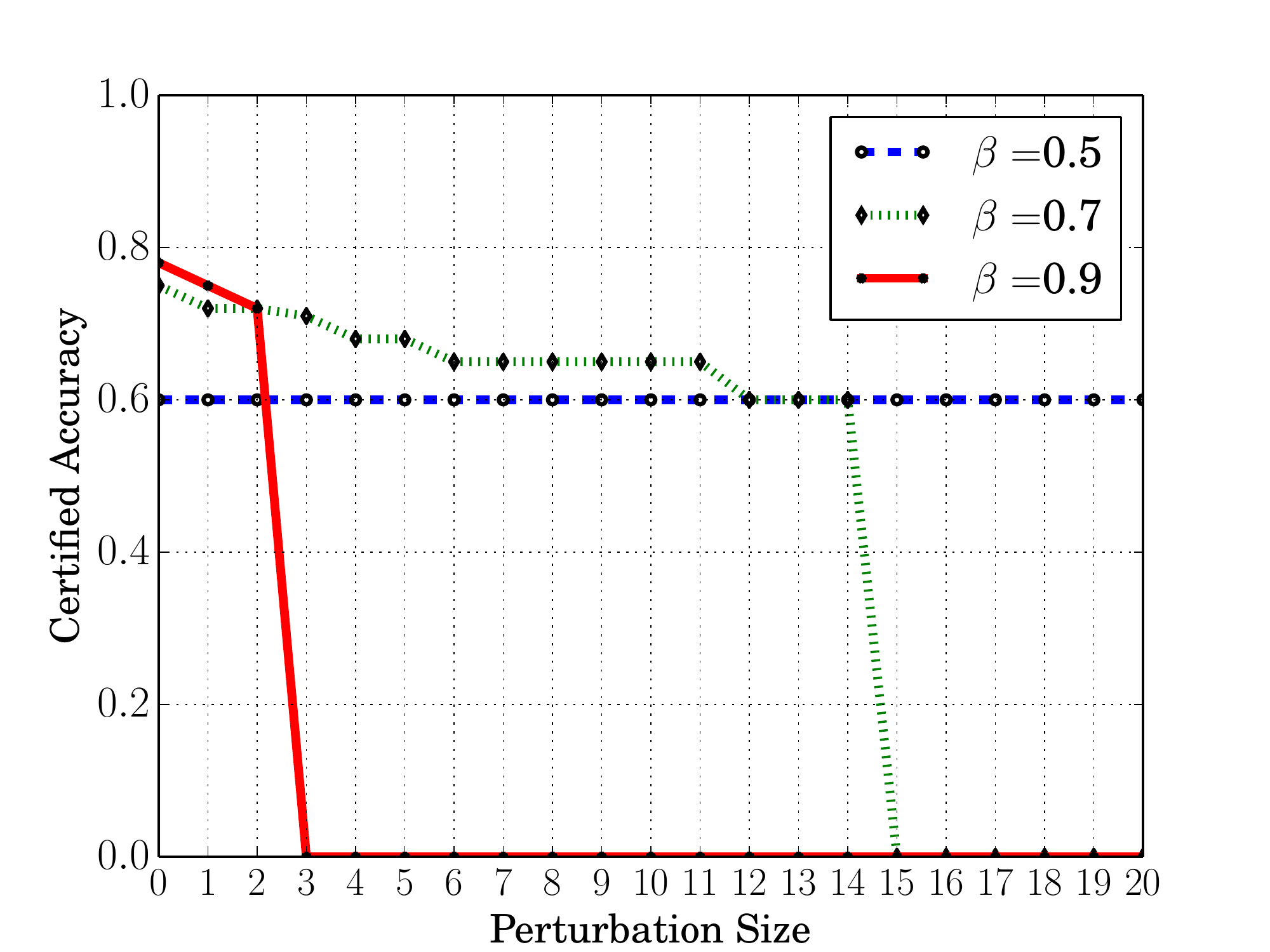} \label{fig3}} 
\vspace{-2mm}
\caption{Impact of $\beta$ on the certified accuracy of the smoothed GAT.}
\label{certify_GAT}
\vspace{-8mm}
\end{figure*}

\begin{figure*}[!t]
\center
\subfloat[MUTAG]{\includegraphics[width=0.28\textwidth]{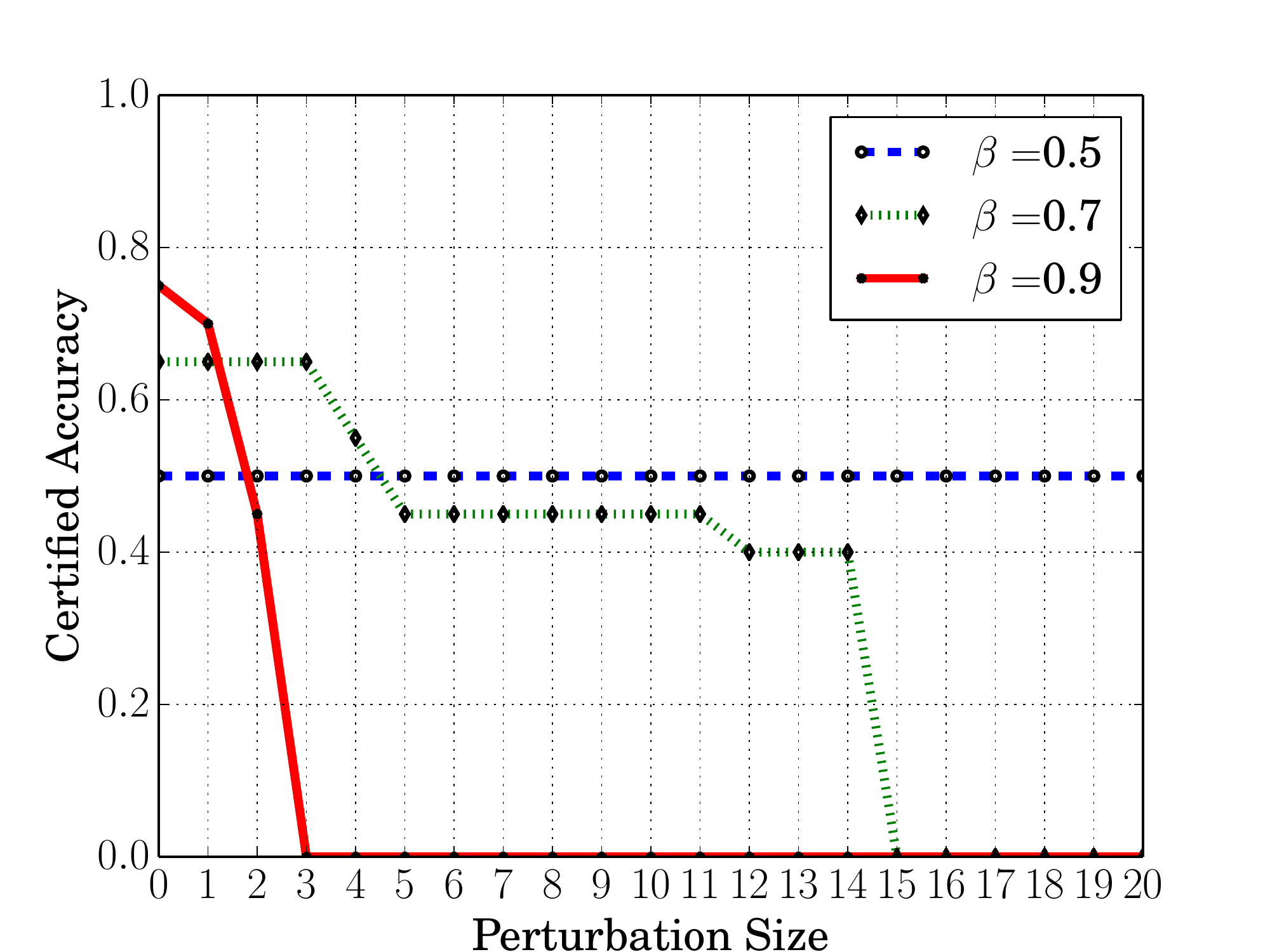} \label{mutag}} 
\subfloat[PROTEINS]{\includegraphics[width=0.28\textwidth]{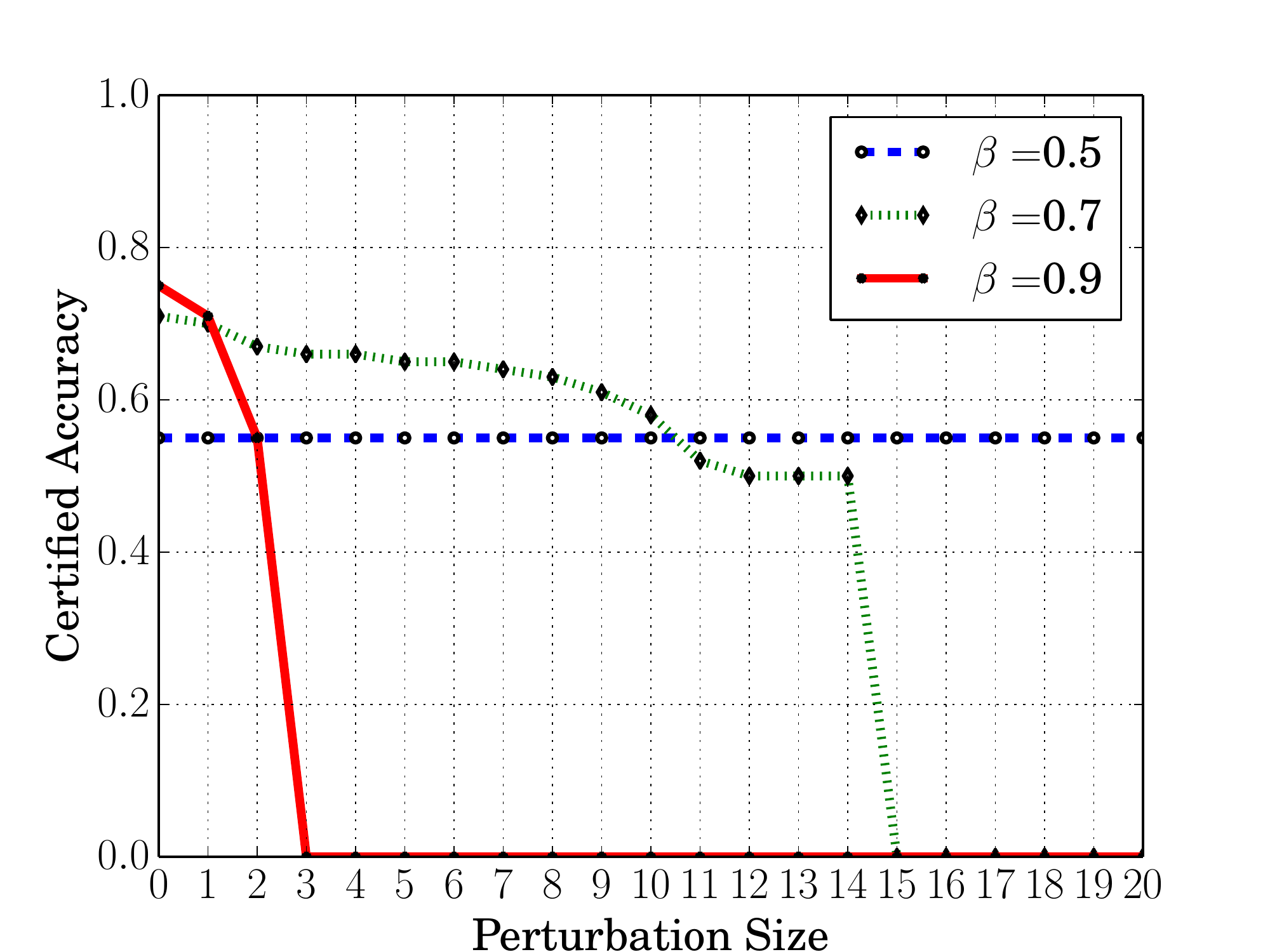} \label{proteins}} 
\subfloat[IMDB]{\includegraphics[width=0.28\textwidth]{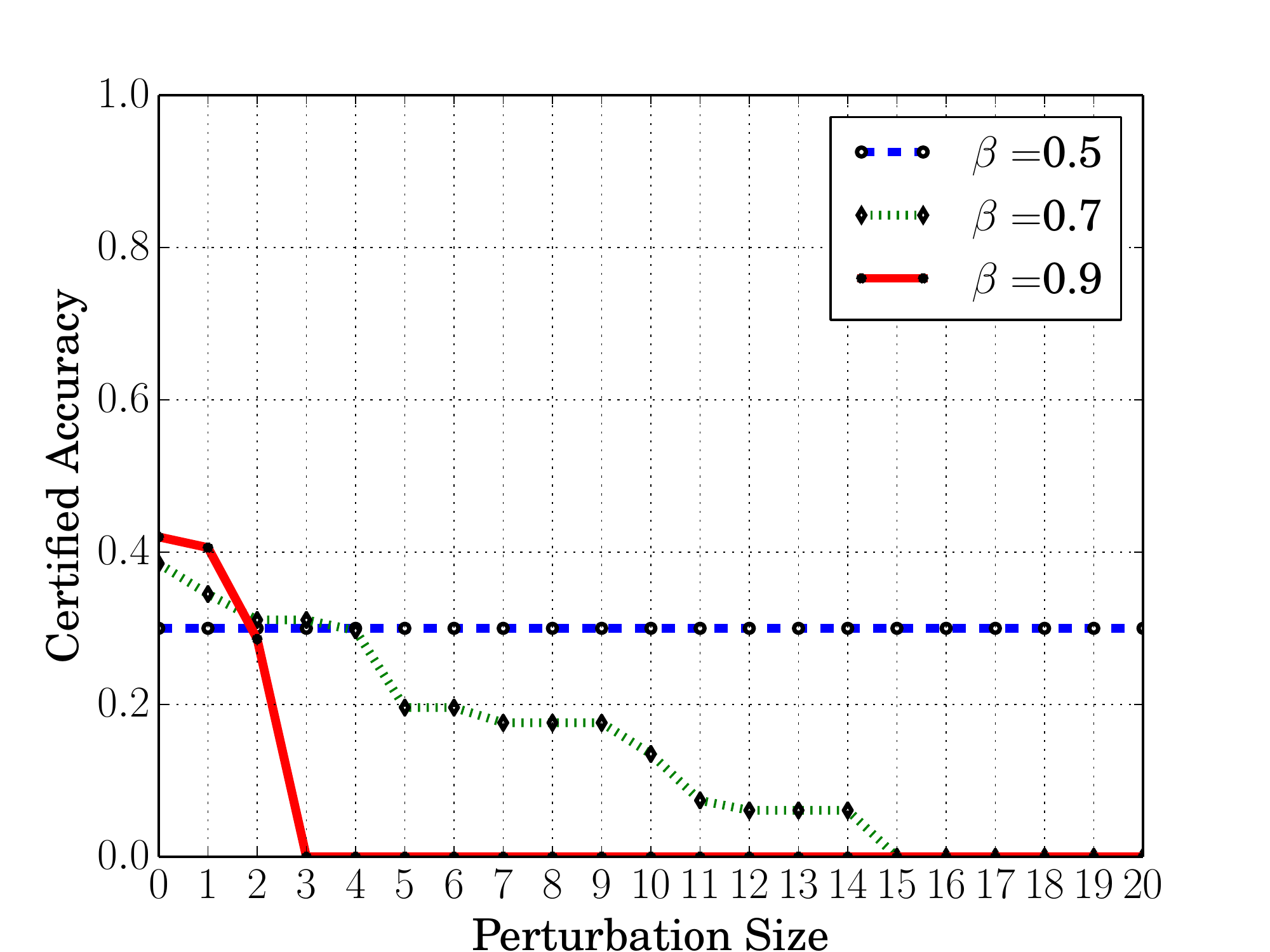} \label{fig3}} 
\vspace{-2mm}
\caption{Impact of $\beta$ on the certified accuracy of the smoothed GIN.}
\label{certify_GIN}
\vspace{-2mm}
\end{figure*}

\subsection{Experimental Setup}

We evaluate  our method on multiple GNNs  and benchmark datasets for both node classification and graph classification. 

\myparatight{Benchmark datasets and GNNs}
We use benchmark graphs and GNNs for both node and graph classification.  Table~\ref{dataset_stat} shows the statistics of our graphs.
\begin{itemize}[leftmargin=*]
\item {\bf Node classification}: We consider Graph Convolutional Network (GCN)~\cite{kipf2017semi} and Graph Attention Network (GAT)~\cite{velivckovic2018graph} for node classification. Moreover, we use the Cora, Citeseer, and Pubmed datasets~\cite{sen2008collective}. They are citation graphs, where nodes are documents and edges indicate citations between documents. In particular,  an undirected edge between two documents is created if one document cites the other. The bag-of-words feature of a document is treated as the node feature. Each document also has a label. 

\item {\bf Graph classification}: We consider Graph Isomorphism Network (GIN)~\cite{xu2019powerful} for graph classification. Moreover, we use the MUTAG, PROTEINS, and  IMDB datasets~\cite{yanardag2015deep}. MUTAG and PROTEINS are bioinformatics datasets. MUTAG contains 188 mutagenic aromatic and heteroaromatic nitro compounds, where each compound represents a graph and each label means whether or not the compound has a mutagenic effect on the Gramnegative bacterium \emph{Salmonella typhimurium}. PROTEINS is a dataset where nodes are secondary structure elements and there is an edge between two nodes if they are neighbors in the amino-acid sequence or in three-dimensional space. Each protein is represented as a graph and is labeled as enzyme or
non-enzyme. IMDB is a movie collaboration dataset. Each graph is an ego-network of an actor/actress, where nodes are actors/actresses and an edge between two actors/actresses indicates that they appear in the same movie. Each graph is obtained from a certain genre of movies, and the task is to classify the genre of a graph.  
\end{itemize}

\myparatight{Training and testing} 
For each node classification dataset, following previous works~\cite{kipf2017semi,velivckovic2018graph}, we sample 20 nodes from each class uniformly at random as the training dataset. Moreover, we randomly sample 100 nodes for testing.   
For each graph classification dataset, we use 90\% of the graphs for training and the remaining 10\%  for testing, similar to~\cite{xu2019powerful}.

\begin{figure}[t]
\center
\subfloat[Smoothed GCN on Cora]{\includegraphics[width=0.23\textwidth]{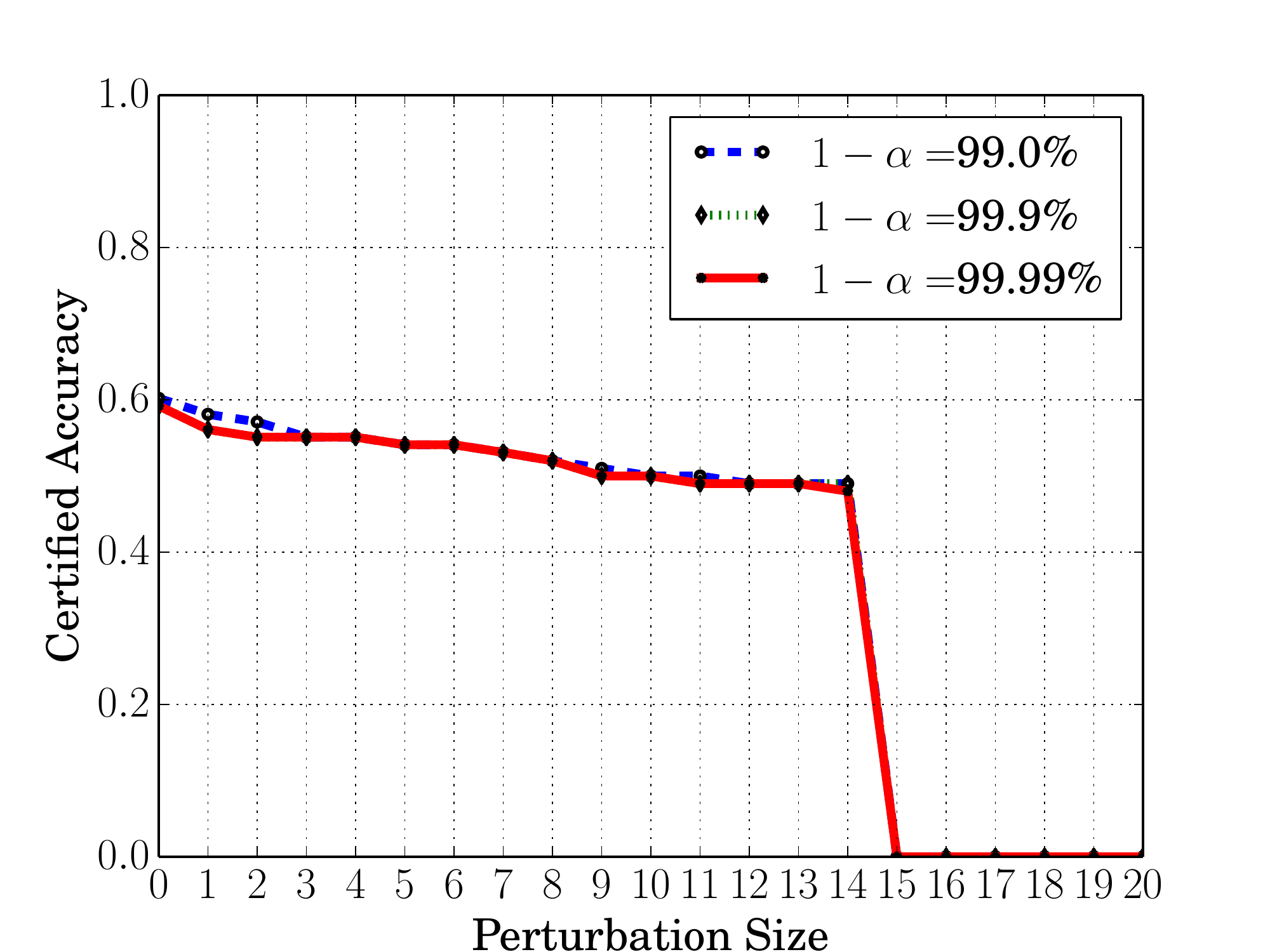} \label{confidence_GCN}}
\subfloat[Smoothed GIN on MUTAG]{\includegraphics[width=0.23\textwidth]{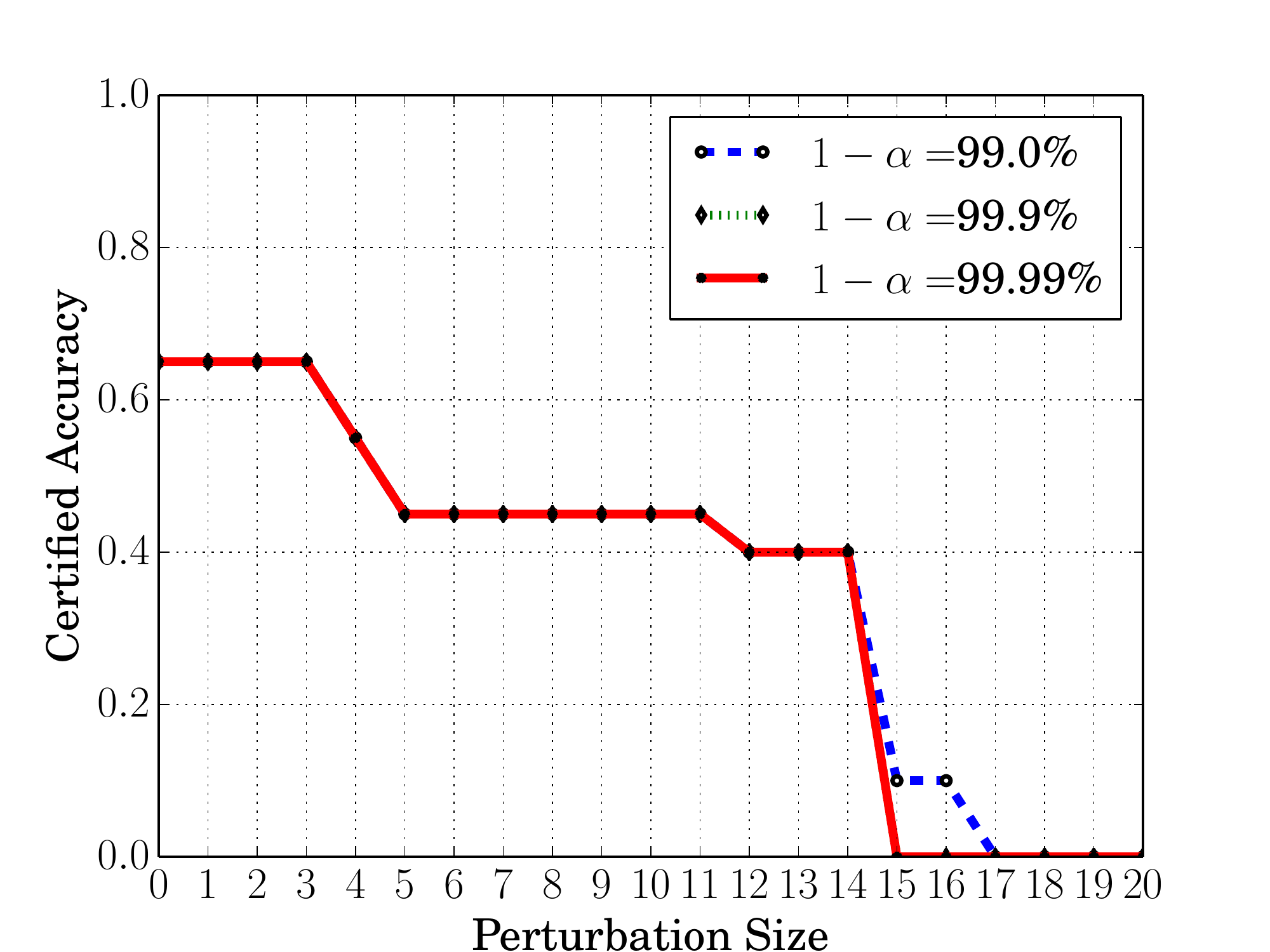} \label{confidence_GIN}} 
\caption{Impact of $1-\alpha$ on the certified accuracy of (a) the smoothed GCN and (b) the smoothed GIN.}
\label{confidence_level}
\end{figure}

 \vspace{-4mm}
\myparatight{Parameter setting} We implement our method in pyTorch. To compute the certified perturbation size,  our method needs to specify the noise parameter $\beta$, the confidence level $1-\alpha$, and the number of samples $d$. 
Unless otherwise mentioned, we set  $\beta=0.7$, $1-\alpha=99.9\%$,  and  $d = 10,000$. We also explore the impact of each parameter while fixing the other parameters to the default settings in our experiments. When computing the certified perturbation size for a node $u$ in node classification, we consider an attacker perturbs the connection status between $u$ and the remaining nodes in the graph. 
We use the publicly available source code for GCN, GAT, and GIN. 
We use our randomized smoothing to smooth each classifier.

\subsection{Experimental Results}
Like Cohen et al.~\cite{cohen2019certified}, we use  \emph{certified accuracy} as the metric to evaluate our method. Specifically,  for a smoothed GNN classifier and a given perturbation size, certified accuracy is the fraction of testing nodes (for node classification) or testing graphs (for graph classification), whose labels are correctly predicted by the smoothed classifier and whose certified perturbation size is no smaller than the given perturbation size.

\myparatight{Impact of the noise parameter $\beta$}
Figure~\ref{certify_GCN} and Figure~\ref{certify_GAT} respectively show the certified accuracy of the smoothed GCN and smoothed GAT vs. perturbation size for different $\beta$ on the three node classification datasets. Figure~\ref{certify_GIN} shows the certified accuracy of GIN with randomized smoothing vs. perturbation size for different $\beta$  on the three graph classification datasets.

We have two observations. First, when $\beta=0.5$, the certified accuracy is independent with the perturbation size, which means that an attacker can not attack a smoothed GNN classifier via perturbing the graph structure. This is because $\beta=0.5$ essentially means that the graph is sampled from the space $\{0,1\}^n$ uniformly at random. In other words, the graph structure does not contain information about the node labels and a smoothed GNN classifier is reduced to only using the node features.  
Second, $\beta$ controls a tradeoff between accuracy under no attacks and robustness. Specifically, when $\beta$ is larger, the accuracy under no attacks (i.e., perturbation size is 0) is larger, but the certified accuracy drops more quickly as the perturbation size increases.

\myparatight{Impact of the confidence level $1-\alpha$}
Figure~\ref{confidence_GCN} and~\ref{confidence_GIN} show the certified accuracy of the smoothed GCN and smoothed GIN vs. perturbation size for different confidence levels, respectively. 
We observe that as the confidence level $1-\alpha$ increases, the certified accuracy curve becomes slightly lower. This is because a higher confidence level leads to a looser estimation of the probability bound $\underline{p_A}$ and $\overline{p_B}$, which means a smaller certified perturbation size for a testing node/graph.  However, the differences of the certified accuracies between different confidence levels are negligible when the confidence levels are large enough.

\myparatight{Impact of the number of samples $d$}
Figure~\ref{number_GCN} and~\ref{number_GIN} show the certified accuracy of the smoothed GCN and smoothed GIN vs. perturbation size for different numbers of samples $d$, respectively. We observe that as the number of samples increases, the certified accuracy curve becomes higher. 
This is because a larger number of samples makes the estimated probability bound $\underline{p_A}$ and $\overline{p_B}$ tighter, which means a larger certified perturbation size for a testing node/graph. On  Cora, the smoothed GCN achieves certified accuracies of 0.55, 0.50, and 0.49 when  $d=50K$ and the attacker arbitrarily adds/deletes at most 5, 10, and 15 edges, respectively. On  MUTAG, GIN with our randomized smoothing achieves certified accuracies of 0.45, 0.45, and 0.40 when $d=50K$ and the attacker arbitrarily adds/deletes at most 5, 10, and 15 edges, respectively.

\begin{figure}[t]
\center
\subfloat[Smoothed GCN on Cora]{\includegraphics[width=0.23\textwidth]{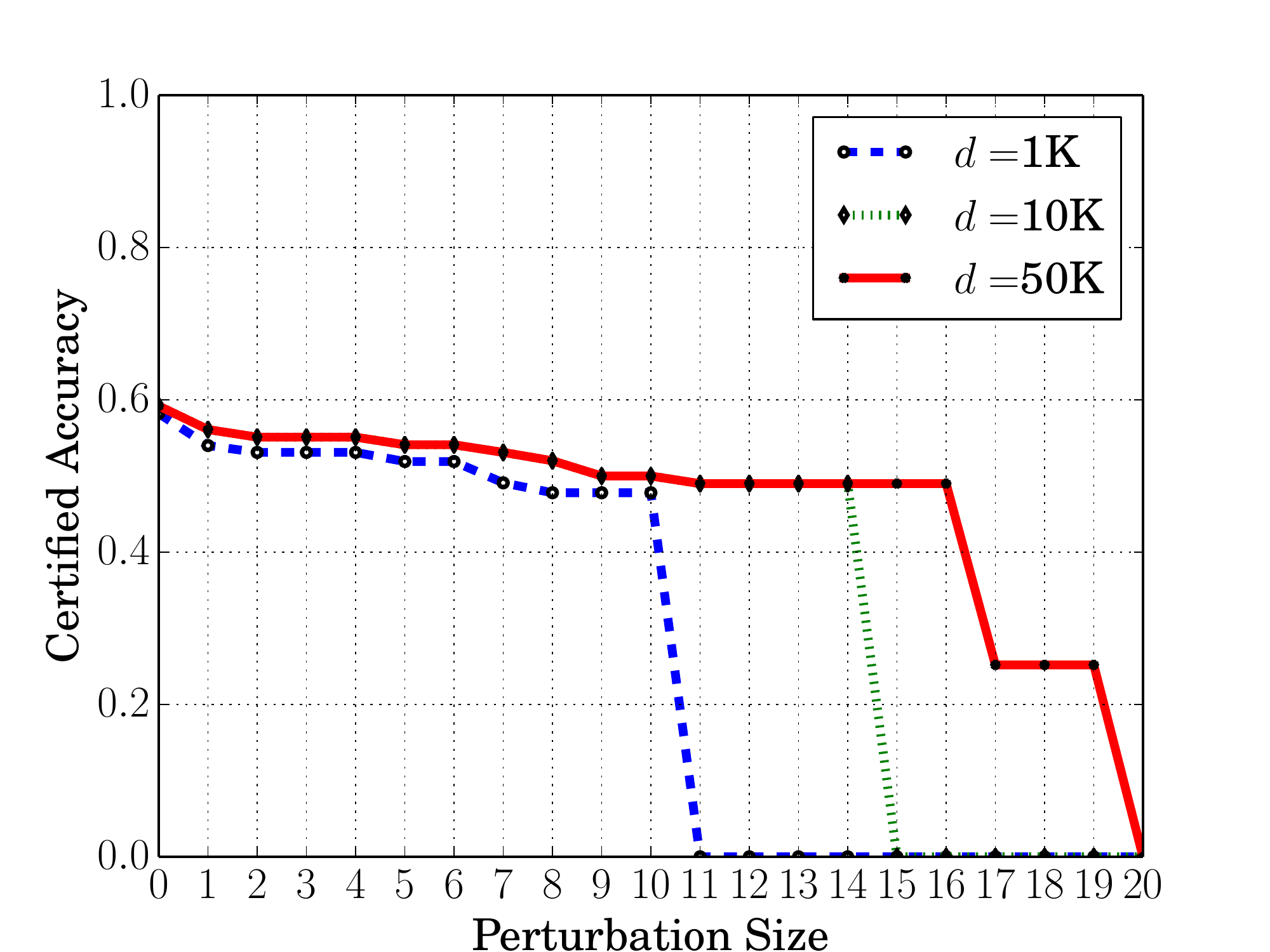} \label{number_GCN}}
\subfloat[Smoothed GIN on MUTAG]{\includegraphics[width=0.23\textwidth]{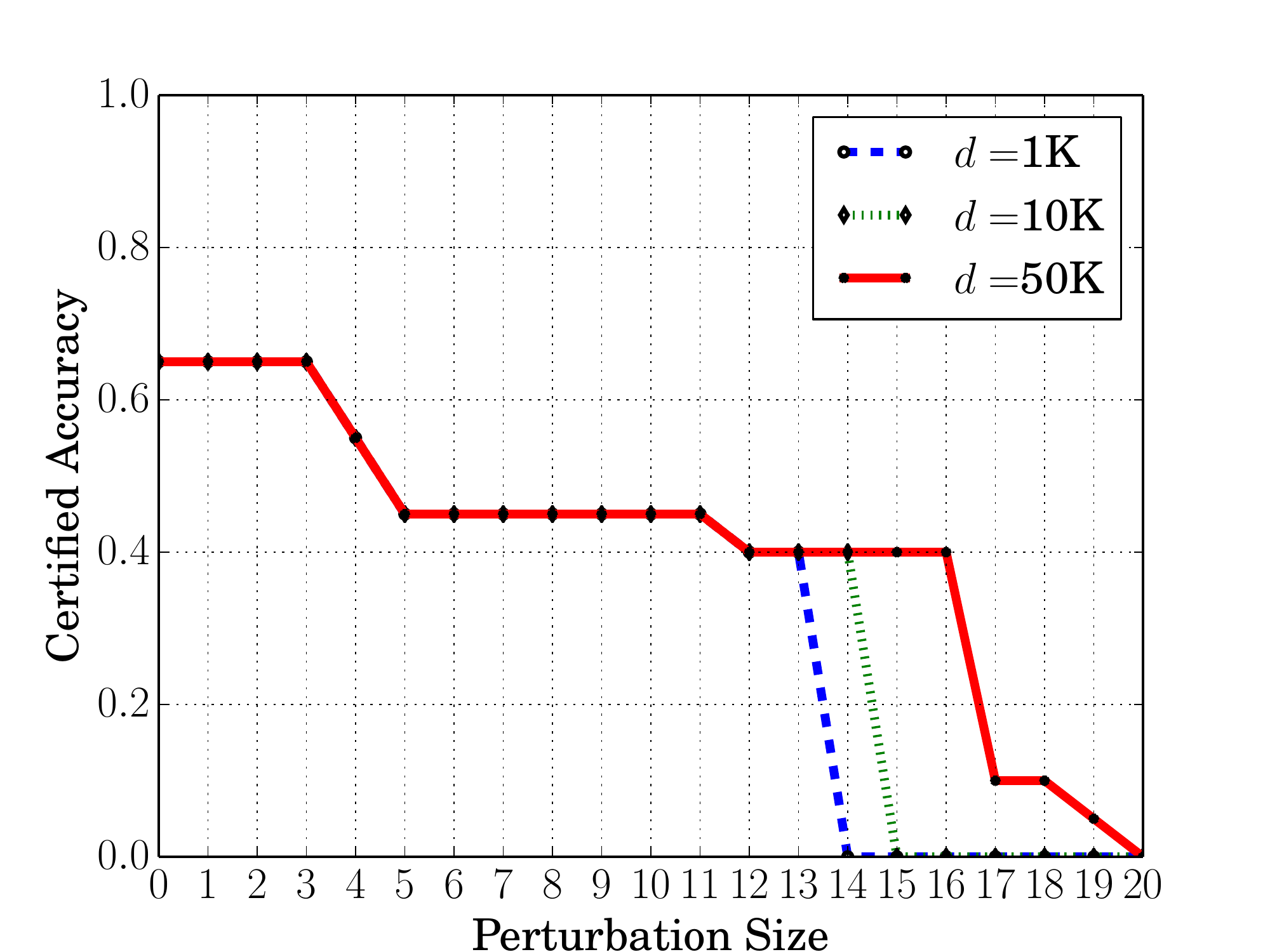} \label{number_GIN}} 
\caption{Impact of $d$ on the certified accuracy of (a) the smoothed GCN and (b) the smoothed GIN. }
\label{number_sample}
\end{figure}

\section{Related Work}
\label{related}

We review studies on certified robustness for classifiers on both non-graph data and graph data.

\subsection{Non-graph Data} 
Various methods have been tried to certify robustness of classifiers. A classifier is certifiably robust if it predicts the same label for data points in a certain region around an input example. Existing certification methods leverage Satisfiability Modulo Theories~\cite{scheibler2015towards,carlini2017provably,ehlers2017formal,katz2017reluplex}, mixed integer-linear programming~\cite{cheng2017maximum,fischetti2018deep,bunel2018unified}, linear programming~\cite{wong2017provable,wong2018scaling}, semidefinite programming~\cite{raghunathan2018certified,raghunathan2018semidefinite}, dual optimization~\cite{dvijotham2018training,dvijotham2018dual}, abstract interpretation~\cite{gehr2018ai2,mirman2018differentiable,singh2018fast}, and layer-wise relaxtion~\cite{weng2018towards,zhang2018efficient}. However, these methods are not scalable to large neural networks and/or are only applicable to  specific neural network architectures.  

Randomized smoothing~\cite{cao2017mitigating,liu2018towards,lecuyer2018certified,li2018second,cohen2019certified,salman2019provably,lee2019stratified,zhai2020macer,levine2020robustness,jia2020certified} was originally developed to defend against adversarial examples. 
It was first proposed as an empirical defense~\cite{cao2017mitigating,liu2018towards}.
For instance,  \citet{cao2017mitigating} proposed to use uniform random noise from a hypercube centered at an input example to smooth a base classifier.  
\citet{lecuyer2018certified}  derived a certified robustness guarantee for randomized smoothing with Gaussian noise or Laplacian noise via differential privacy techniques.  
\citet{li2018second} leveraged information theory to derive a tighter certified robustness guarantee. 
\citet{cohen2019certified} leveraged the Neyman-Pearson Lemma~\cite{neyman1933ix} to obtain a tight certified robustness guarantee under $L_2$-norm for randomized smoothing with  Gaussian noise. 
Specifically, they showed that a smoothed classifier verifiably predicts the same label when the $L_2$-norm of the adversarial perturbation to an input example is less than a threshold. 
\citet{salman2019provably} designed an adaptive attack against smoothed classifiers and used the attack to train classifiers in an adversarial training paradigm to enlarge the certified robuseness under $L_2$-norm.
\citet{zhai2020macer} proposed to explicitly maximize the certified robustness via a new training strategy. 
All these 
methods focused on top-$1$ predictions. 
\citet{jia2020certified} extended \cite{cohen2019certified} to derive the first certification of top-$k$ predictions against adversarial examples.
Compared to other methods for certified robustness, randomized smoothing has two key advantages: 1) it is scalable to large neural networks, and 2) it is applicable to any base classifier.

\subsection{Graph Data} 
Several studies~\cite{dai2018adversarial,zugner2018adversarial,zugner2019adversarial,bojchevski2019adversarial,wang2019attacking,wu2019adversarial,xu2019topology,sun2020adversarial,chang2020restricted} have shown that attackers can fool GNNs via manipulating the node features and/or graph structure.
~\cite{wu2019adversarial,xu2019topology,zhu2019robust,tang2020transferring,entezari2020all,tao2021adversarial} proposed empirical defenses without certified robustness guarantees. 
A few recent works~\cite{Zugner2019Certifiable,bojchevski2020efficient,jin2020certified,bojchevski2019certifiable,zugner2020certifiable}
study certified robustness of GNNs.  \citet{Zugner2019Certifiable}
considered node feature perturbation. In particular, they derived certified robustness for a particular type of GNN called graph convolutional network~\cite{kipf2017semi} against  node feature perturbation. More specifically, they formulated the certified robustness as a linear programming problem, where the objective is to require that a node's prediction is constant within an allowable node feature perturbation. Then, they derived the robustness guarantee based on the dual form of the optimization problem, which is motivated by~\cite{wong2017provable}. 
\cite{bojchevski2019certifiable,zugner2020certifiable,jin2020certified}
considered structural perturbation for a specific GNN. 
Specifically, \cite{bojchevski2019certifiable} required that GNN prediction is a linear function of (personalized) PageRank~\cite{klicpera2018predict} and  \cite{zugner2020certifiable,jin2020certified} could only certify the robustness of GCNs. For example, \cite{jin2020certified} derived the certified robustness of GCN for graph classification via utilizing dualization and convex envelope tools.

Two recent works~\cite{jia2020certifiedcommunity,bojchevski2020efficient} leveraged randomized smoothing to provide a model-agnostic certification 
against structural perturbation. 
Jia et al.~\cite{jia2020certifiedcommunity} certified the robustness of community detection against structural perturbation. 
They essentially model the problem as  binary classification, i.e., whether a set of nodes are in the same community or not, and design a randomized smoothing technique to certify its robustness. 
Bojchevski et al.~\cite{bojchevski2020efficient} generalize the randomized smoothing technique
developed in [29] to the sparse setting and derive certified robustness
for GNNs. They did not formally prove that their multi-class certificate
is tight. However, in the special case of equal flip probabilities,
which is equivalent to our certificate, our tightness proof (Theorem 2)
applies.

\section{Conclusion and Future Work}
\label{conclusion}

In this work, we develop the certified robustness guarantee of GNNs for both node and graph classifications against structural perturbation. Our results are applicable to any GNN classifier. Our certification is based on randomized smoothing for binary data which we develop in this work. Moreover, we prove that our certified robustness guarantee is tight when randomized smoothing is used and no assumptions on the GNNs are made. An interesting future work is to incorporate the information of a given GNN to further improve the certified robustness guarantee. 
\section{Proofs}

\subsection{Proof of Lemma~\ref{lemmabounds}}
\label{proof_lemmabounds}

We first describe the Neyman-Pearson Lemma for binary random variables, which we use to prove Lemma~\ref{lemmabounds}. 

\begin{customLemma}{2}[Neyman-Pearson Lemma for Binary Random Variables]
\label{neyman-pearson-variant}
Let $X$ and $Y$ be two random variables in the discrete space $\{0,1\}^{n}$ with probability distributions $\textrm{Pr}(X)$ and $\textrm{Pr}(Y)$, respectively. Let $h: \{0,1\}^{n} \rightarrow \{0, 1\}$ be a random or deterministic function. 
\begin{itemize}
\item Let $S_1 = \{ \mathbf{z} \in \{0, 1\}^n: \frac{\textrm{Pr}(X=\mathbf{z})}{\textrm{Pr}(Y=\mathbf{z})} > r \}$ 
and 
$S_2 = \{ \mathbf{z} \in \{0, 1\}^n: \frac{\textrm{Pr}(X=\mathbf{z})}{\textrm{Pr}(Y=\mathbf{z})} = r \}$ for some $r>0$.
Assume $S_3 \subseteq S_2$ and $S = S_1 \bigcup S_3$. If $\textrm{Pr}(h(X)=1) \geq \textrm{Pr}(X \in S)$, then $\textrm{Pr}(h(Y)=1) \geq \textrm{Pr}(Y \in S)$.
\item Let $S_1  = \{ \mathbf{z} \in \{0, 1\}^n: \frac{\textrm{Pr}(X=\mathbf{z})}{\textrm{Pr}(Y=\mathbf{z})} < r \}$ 
and  
$S_2 = \{ \mathbf{z} \in \{0, 1\}^n: \frac{\textrm{Pr}(X=\mathbf{z})}{\textrm{Pr}(Y=\mathbf{z})} = r \}$ for some $r>0$.
Assume $S_3 \subseteq S_2$ and $S = S_1 \bigcup S_3$. If $\textrm{Pr}(h(X)=1) \leq \textrm{Pr}(X \in S)$, then $\textrm{Pr}(h(Y)=1) \leq \textrm{Pr}(Y \in S)$.
\end{itemize}
\end{customLemma}
\begin{proof}
The proof can be found in a standard statistics textbook, e.g.,~\cite{lehmann2006testing}. For completeness, we include the proof here.
Without loss of generality, we assume that $h$ is random and denote $h(1|z)$ (resp. $h(0|z)$) as the probability that $h(z)=1$ (resp. $h(z)=0)$).
We denote $S^c$ as the complement of $S$, i.e., $S^c=\{0, 1\}^n\setminus S$.  For any $\mathbf{z} \in S$, we have $\frac{\textrm{Pr}(X=\mathbf{z})}{\textrm{Pr}(Y=\mathbf{z})} \geq r$, and for any $\mathbf{z} \in S^c$,  we have $\frac{\textrm{Pr}(X=\mathbf{z})}{\textrm{Pr}(Y=\mathbf{z})} \leq r$.
We  prove the first part, and the second part can be proved similarly. 
{
\begin{align*}
& \textrm{Pr}(h(Y)=1) - \textrm{Pr}(Y \in S) = \\
& \sum_{\mathbf{z} \in \{0,1\}^{n}} h(1|\mathbf{z}) \textrm{Pr}(Y = \mathbf{z}) - \sum_{\mathbf{z} \in S} \textrm{Pr}(Y = \mathbf{z}) \\
& = \Big( \sum_{\mathbf{z} \in S^{c}} h(1|\mathbf{z}) \textrm{Pr}(Y = \mathbf{z}) + \sum_{\mathbf{z} \in S} h(1|\mathbf{z}) \textrm{Pr}(Y = \mathbf{z}) \Big) \nonumber \\ 
&   \qquad - \Big( \sum_{\mathbf{z} \in S} h(1|\mathbf{z}) \textrm{Pr}(Y = \mathbf{z}) + \sum_{\mathbf{z} \in S} h(0|\mathbf{z}) \textrm{Pr}(Y = \mathbf{z}) \Big) \\
& = \sum_{\mathbf{z} \in S^{c}} h(1|\mathbf{z}) \textrm{Pr}(Y = \mathbf{z}) - \sum_{\mathbf{z} \in S} h(0|\mathbf{z}) \textrm{Pr}(Y = \mathbf{z}) \\
& \geq \frac{1}{r} \Big( \sum_{\mathbf{z} \in S^{c}} h(1|\mathbf{z}) \textrm{Pr}(X = \mathbf{z}) - \sum_{\mathbf{z} \in S} h(0|\mathbf{z}) \textrm{Pr}(X = \mathbf{z}) \Big) \\
& = \frac{1}{r} \Bigg( \Big( \sum_{\mathbf{z} \in S^{c}} h(1|\mathbf{z}) \textrm{Pr}(X = \mathbf{z}) + \sum_{\mathbf{z} \in S} h(1|\mathbf{z}) \textrm{Pr}(X = \mathbf{z}) \Big) \nonumber \\ 
& \qquad - \Big( \sum_{\mathbf{z} \in S} h(1|\mathbf{z}) \textrm{Pr}(X = \mathbf{z}) + \sum_{\mathbf{z} \in S} h(0|\mathbf{z}) \textrm{Pr}(X = \mathbf{z}) \Big) \Bigg) \\ 
& = \frac{1}{r} \Big( \sum_{\mathbf{z} \in \{0,1\}^{n}} h(1|\mathbf{z}) \textrm{Pr}(X = \mathbf{z}) - \sum_{\mathbf{z} \in S} \textrm{Pr}(X = \mathbf{z}) \Big) \\
& = \frac{1}{r} \Big( \textrm{Pr}(h(X)=1) - \textrm{Pr}(X \in S) \Big) \\
& \geq 0.
\end{align*}
}%
\end{proof}

Next, we restate Lemma~\ref{lemmabounds} and show our proof.
\certifyiedradiuslemma*

\begin{proof}
Based on the regions $\mathcal{A}$ and $\mathcal{B}$ defined in Equation~\ref{region_A} and Equation~\ref{region_B}, we have  $\text{Pr}(X \in \mathcal{A}) = \underline{p_A}$ and $\text{Pr}(X \in \mathcal{B}) = \overline{p_B}$. 

Moreover, based on the conditions in Equation~\ref{main_theorem_condition_label}, we have:
$$\text{Pr}(f(X)=c_A)\geq \underline{p_A} = \text{Pr}(X\in\mathcal{A});$$ 
$$\text{Pr}(f(X)=c_B)\leq \overline{p_B} =\text{Pr}(X\in\mathcal{B}).$$ 

Next, we leverage Lemma~\ref{neyman-pearson-variant} to derive the condition for $\text{Pr}(f(Y)=c_A)> \text{Pr}(f(Y)=c_B)$. 
Specifically, we define $h(\mathbf{z})=\mathbb{I}(f(\mathbf{z})=c_A)$. Then, we have: 
$$\text{Pr}(h(X)=1) = \text{Pr}(f(X)=c_A)\geq \text{Pr}(X\in\mathcal{A}).$$ 
Moreover,  we have $\frac{\text{Pr}(X=\mathbf{z})}{\text{Pr}(Y=\mathbf{z})} > r_{a^{\star}}$ for any $\mathbf{z} \in \bigcup_{j=1}^{a^{\star}-1} \mathcal{R}_j$ and 
$\frac{\text{Pr}(X=\mathbf{z})}{\text{Pr}(Y=\mathbf{z})} = r_{a^{\star}}$ for any $\mathbf{z} \in \underline{\mathcal{R}_{a^\star}}$, where $r_{a^{\star}}$ is the probability density ratio in the region $\mathcal{R}_{a^{\star}}$. Therefore, according to the first part of Lemma~\ref{neyman-pearson-variant}, we have $ \text{Pr}(f(Y)=c_A) \geq \text{Pr}(Y\in \mathcal{A})$. 
Similarly, based on the second part of Lemma~\ref{neyman-pearson-variant}, we have
$\text{Pr}(f(Y)=c_B)\leq \text{Pr}(Y\in \mathcal{B})$.
\end{proof}

\subsection{Proof of Theorem~\ref{CertifiedPerturbationSize}}
\label{proof_theorem}

Recall that our goal is to make $\text{Pr}(f(Y) = c_A) > \text{Pr}(f(Y)= c_B)$. 
Based on Lemma~\ref{lemmabounds}, it is sufficient to require 
$
\text{Pr}(Y \in \mathcal{A}) > \text{Pr}(Y \in \mathcal{B}).
$
Therefore, we derive the certified perturbation size $K$ as the maximum $k$ such that the above inequality holds for $\forall ||\delta||_0=R$.

\subsection{Proof of Theorem~\ref{tighnessbound}}
\label{proof_tight}

Our idea is to construct a base classifier $f^*$ consistent with the conditions in Equation~\ref{main_theorem_condition_label}, but the smoothed classifier is not guaranteed to predict $c_A$. 
Let disjoint regions $\mathcal{A}$ and $\mathcal{B}$ be defined as in Equation~\ref{region_A} and Equation~\ref{region_B}.
We denote $\Gamma=\{1,2,\cdots,c\}\setminus\{c_A,c_B\}$. Then, for each label $c_i$ in $\Gamma$, we can find a region $\mathcal{C}_{i}\subseteq\{0,1\}^n \setminus(\mathcal{A}\cup\mathcal{B})$ such that $\mathcal{C}_{i}\cap\mathcal{C}_{j}=\emptyset,\forall i,j\in\Gamma,i\neq j$ and $\text{Pr}(X\in \mathcal{C}_i)\leq \overline{p_B}$. We can construct such regions because $\underline{p_A}+(|\mathcal{C}|-1)\cdot\overline{p_B}\geq 1$. 
Given those regions, we construct the following base classifier:
\begin{align*}
    f^{*}(\mathbf{z})=
    \begin{cases}
    c_A  & \text{if }\mathbf{z}\in\mathcal{A} \\
    c_B & \text{if }\mathbf{z}\in\mathcal{B}\\
    c_i & \text{if }\mathbf{z}\in \mathcal{C}_{i}
    \end{cases}
\end{align*}
By construction, we have $\text{Pr}(f^{*}(X)=c_A)=\underline{p_A}$,  $\text{Pr}(f^{*}(X)=c_B)=\overline{p_B}$, and $\text{Pr}(f^{*}(X)=c_i) \leq \overline{p_B}$ for any $c_i \in \Gamma $, which are consistent with Equation~\ref{main_theorem_condition_label}. 
From our proof of Theorem~\ref{CertifiedPerturbationSize}, we know that when $||\delta||_0 > K$, we have:
\begin{align*}
    \text{Pr}(Y\in\mathcal{A}) \leq \text{Pr}(Y\in\mathcal{B}),
\end{align*}
or equivalently we have:
\begin{align*}
\text{Pr}(f^{*}(Y)=c_A) \leq 
\text{Pr}(f^{*}(Y)=c_B)
\end{align*}
Therefore, we have either $g(\mathbf{s} \oplus \delta)\neq c_A$ or there exist ties when $||\delta||_0 > K$.

\subsection{Computing $\text{Pr}(X \in \mathcal{R}(m))$ and $\text{Pr}(Y \in \mathcal{R}(m))$}
\label{compute_prob}

We first define the following regions:
\begin{align*}
\mathcal{R}(w,v) & = \{\mathbf{z} \in \{0, 1\}^n: ||\mathbf{z}-\mathbf{s}||_0=w \text{ and } ||\mathbf{z}-\mathbf{s} \oplus \delta||_0 =  v \},
\end{align*}
for $w,v \in \{0, 1, \cdots, n \}$. Intuitively, $\mathcal{R}(w,v)$ includes the binary vectors that are $w$ bits different from $\mathbf{s}$ and $v$ bits different from 
$\mathbf{s} \oplus \delta$.  
 Next, we compute the size of the region $\mathcal{R}(w,v)$ when $||\delta||_0=k$. Without loss of generality, we assume $\mathbf{s} = [0, 0, \cdots, 0]$ as a zero vector and $\mathbf{s} \oplus \delta = [1, 1, \cdots, 1, 0, 0, \cdots, 0]$, where the first $k$ entries are 1 and the remaining $n - k$ entries are 0. We construct a binary vector $\mathbf{z}\in \mathcal{R}(w,v)$. Specifically,  suppose we flip $i$ zeros in the last $n-k$ zeros in both $\mathbf{s}$ and $\mathbf{s} \oplus \delta$. Then, we flip $w-i$ of the first $k$ bits of  $\mathbf{s}$ and flip the rest $k-w+i$ bits of the first $k$ bits of  $\mathbf{s} \oplus \delta$. In order to have $\mathbf{z}\in \mathcal{R}(w,v)$, we need  $k-w+i + i = v$, i.e., $i=(w+v-k)/2$. Therefore, we have the size  $|\mathcal{R}(w,v)|$ of the region $\mathcal{R}(w,v)$ as follows:
   {
  \begin{align*}
  |\mathcal{R}(w,v)|  = 
    \begin{cases}
      0,  &\textrm{ if } (w+v-k) \textrm{ mod } 2 \neq 0, \\
      0,  &\textrm{ if } w+v<k, \\
  {n-k \choose \frac{w+v-k}{2}} {k \choose \frac{w-v+k}{2}}, &\textrm{ otherwise }  
  \end{cases}
  \end{align*}
  }%
Moreover, for each $\mathbf{z} \in \mathcal{R}(w,v)$, we have $\text{Pr}(X = \mathbf{z}) = \beta^{n-w} (1-\beta)^{w}$  and $\text{Pr}(Y = \mathbf{z}) = \beta^{n-v} (1-\beta)^v$. Therefore, we have:
  {
  \begin{align*}
  \textrm{Pr}(X \in \mathcal{R}(w,v)) 
  & = \beta^{n-w} (1-\beta)^w \cdot |\mathcal{R}(w,v)|, \\ 
  \textrm{Pr}(Y \in \mathcal{R}(w,v)) 
  & = \beta^{n-v} (1-\beta)^v \cdot |\mathcal{R}(w,v)|. 
  \end{align*}
  }%

Note that $\mathcal{R}(m)=\cup_{v-w=m} \mathcal{R}(w, v)$. 
 Therefore, we have: 
  {
  \begin{align*}
  & \text{Pr}(X \in \mathcal{R}(m)) \\
  =& \text{Pr}(X \in \cup_{v-w=m} \mathcal{R}(w, v)) \\
  =  & \text{Pr}(X \in \cup_{j=\max(0,m)}^{\min(n,n+m)} \mathcal{R}(j-m,j))  \\ 
  =  & \sum_{j=\max(0,m)}^{\min(n,n+m)} \text{Pr}(X \in \mathcal{R}(j-m,j))  \\ 
  =  & \sum_{j=\max(0,m)}^{\min(n,n+m)} \beta^{n-(j-m)} (1-\beta)^{j-m} \cdot |\mathcal{R}(j-m,j)| \label{probX} \\
    =  & \sum_{j=\max(0,m)}^{\min(n,n+m)} \beta^{n-(j-m)} (1-\beta)^{j-m} \cdot t(m,j),
  \end{align*}
  }%
%
where $t(m,j)=|\mathcal{R}(j-m,j)|$.
%

Similarly, we have
\begin{align*}
& \text{Pr}(Y \in \mathcal{R}(m)) = \sum_{j=\max(0,m)}^{\min(n,n+m)} \beta^{n-j} (1-\beta)^{j} \cdot t(m,j).
\end{align*}

\noindent \textbf{Acknowledgements.}
We thank the anonymous reviewers for their constructive
comments. This work is supported by the National Science
Foundation under Grants No. 1937786 and 1937787 and the Army Research Office under Grant No. W911NF2110182. Any opinions, findings and conclusions or
recommendations expressed in this material are those of the
author(s) and do not necessarily reflect the views of the funding
agencies.

\bibliographystyle{ACM-Reference-Format}
\bibliography{refs}

\end{document}